\documentclass[11pt,reqno]{amsart}

\usepackage{amsmath,amssymb,amsthm}
\usepackage{bm}
\usepackage{natbib}
\usepackage{graphicx}
\usepackage{here}
\usepackage{multirow}

\makeatletter

\@addtoreset{equation}{section}
\makeatletter

\newtheorem{theorem}{Theorem}[section]

\newtheorem{lemma}{Lemma}[section]
\newtheorem{proposition}{Proposition}[section]
\newtheorem{assumption}{Assumption}[section]
\theoremstyle{definition}
\newtheorem{definition}{Definition}[section]
\newtheorem{remark}{Remark}[section]

\newcommand{\R}{\mathbb{R}}

\renewcommand{\tilde}{\widetilde}
\renewcommand{\hat}{\widehat}

\DeclareMathOperator{\Var}{Var}
\DeclareMathOperator{\Cov}{Cov}

\begin{document}

\title[]{Nonparametric inference on L\'evy measures of compound Poisson-driven Ornstein-Uhlenbeck processes under macroscopic discrete observations}
\author[D. Kurisu]{Daisuke Kurisu}

\date{This version: \today}

\address[D. Kurisu]{
Department of Industrial Engineering and Economics, School of Engineering, Tokyo Institute of Technology\\
2-12-1 Ookayama, Meguro-ku, Tokyo 152-8552, Japan.
}
\email{kurisu.d.aa@m.titech.ac.jp}

\begin{abstract}
This study examines a nonparametric inference on a stationary L\'evy-driven Ornstein-Uhlenbeck (OU) process $X = (X_{t})_{t \geq 0}$ with a compound Poisson subordinator. We propose a new spectral estimator for the L\'evy measure of the L\'evy-driven OU process $X$ under macroscopic observations. We also derive, for the estimator, multivariate central limit theorems over a finite number of design points, and high-dimensional central limit theorems in the case wherein the number of design points increases with an increase in the sample size. Built on these asymptotic results, we develop methods to construct confidence bands for the L\'evy measure and propose a practical method for bandwidth selection. 
\medskip

\end{abstract}

\maketitle
\section{Introduction}

Given a positive number $\lambda$ and an increasing L\'evy process $J = (J_{t})_{t \geq 0}$ without drift component, an Ornstein-Uhlenbeck (OU) process $X = (X_{t})_{t \geq 0}$ driven by $J$ is defined by a solution to the following stochastic differential equation (SDE)
\begin{align}\label{OU}
dX_t = -\lambda X_tdt + dJ_{\lambda t},\ t \geq 0.
\end{align}
We refer to \cite{Sa99} and \cite{Be96} as standard references on L\'evy processes. 
In this study, we consider a nonparametric inference on the L\'evy measure $\nu$ of the back-driving L\'evy process $J$ in (\ref{OU}) from discrete observations of $X$. The L\'evy measure $\nu$ is defined as a Borel measure on $[0,\infty)$ such that
\[
\int_{0}^{\infty}(1 \wedge x^{2})\nu(dx)<\infty. 
\]
We assume that $X$ is stationary. If $\int_{(2, \infty)}\log x \nu(dx)<\infty$, then the unique stationary solution of (\ref{OU}) exists (see Theorem 17.5 and Corollary 17.9 in \cite{Sa99}), and the stationary distribution $\pi$ of $X$ is self-decomposable with the characteristic function
\begin{align}\label{ch_stmeas}
\varphi(t) = \int_{\mathbb{R}}e^{itx}\pi(dx) = \exp{\left(\int_{0}^{\infty}(e^{itx}-1){k(x) \over x}dx\right)},
\end{align}
where $k(x) = \nu((x,\infty))1_{[0,\infty)}$.

This study focuses on the case wherein the L\'evy process $J$ in (\ref{OU}) is a compound Poisson process. In other words, $J$ is of the form
\[
J_{t} = \sum_{j=1}^{N_{t}}U_{j},\ t \geq 0,
\]
where $N = (N_{t})_{t \geq 0}$ is a Poisson process with intensity $\alpha>0$ and $\{U_{j}\}_{j \geq 1}$ is a sequence of independent and identically distributed (i.i.d.) positive-valued random variables with common distribution $F$. In this case, $J_{t}$ has a characteristic function of the form
\[
\varphi_{J_{t}}(u) = E[e^{iuJ_{t}}] = \exp\left(t\alpha \int_{0}^{\infty}(e^{iux} -1)F(dx)\right)
\]
and the L\'evy measure is given by $\nu(dx) = \alpha F(dx)$. We also work with the macroscopic observation
set up, that is, we have discrete observations $X_{\Delta}, X_{2\Delta}, \hdots, X_{n\Delta}$ at frequency $1/\Delta>0$ with $\Delta = \Delta_{n} \to \infty$ and $\Delta_{n}/n \to 0$ as $n \to \infty$. This is a technical condition to make the dependence among observations $\{X_{j\Delta}\}_{j=1}^{n}$ asymptotically negligible.

This study aims to develop a nonparametric inference on the L\'evy measure of a L\'evy-driven OU process. Therefore, we first propose a spectral (or Fourier-based) estimator for the $k$-function and derive a multivariate central limit theorem for the estimator over finite design points. As an extension of the result, we also derive high-dimensional central limit theorems for the estimator in the case wherein design points over a compact interval included in $(0,\infty)$ increases as the sample size $n$ goes to infinity. Second, built on those limit theorems, we develop methods for implementing confidence bands for the $k$-function. Similar methods to construct ``asymptotic'' uniform confidence bands are also proposed in \cite{HoLe12}. 
Since confidence bands provide a simple graphical description of the accuracy of a nonparametric curve estimator, quantifying uncertainties of the estimator simultaneously over design points, they are practically important in statistical analysis. Third, we propose a practical method for bandwidth selection inspired by the idea developed by \cite{BiDuHoMu07} on bandwidth selection in density deconvolution. To the best of our knowledge, this is the first paper to establish limit theorems for nonparametric estimators for the L\'evy measure of compound Poisson-driven OU processes.

L\'evy-driven OU processes are widely used in modeling phenomena where random events occur at random discrete times. For example, refer to \cite{AlTeTi01}, \cite{KeSt01}, and \cite{NoVeGa14} for applications of these processes to insurance, dam theory, and rainfall models. 
Several authors investigate the parametric inference on L\'evy-driven OU processes driven by subordinators. 
We refer to \cite{HuLo09}, \cite{Ma10}, and \cite{Ma14} under the high-frequency set up (i.e., $\Delta = \Delta_{n} \to 0$ and $n\Delta_{n} \to \infty$ as $n \to \infty$) and \cite{BrDaYa07} under the low-frequency set up (i.e., $\Delta>0$ is fixed and $n \to \infty$). 
There are several studies on parametric and nonparametric estimations and inferences on L\'evy processes. 
We refer to recent contributions by \cite{Wo01}, \cite{KaMa11, KaMa13}, and \cite{BrMa18} on parametric inference on L\'evy processes. 
We also find an overview of recent developments on the parametric inference on L\'evy processes in \cite{Ma15}. 
Some authors have studied statistical inference on L\'evy process under macroscopic observations. 
\cite{DuHo11} investigates statistical inference on a compound Poisson process under three kinds of time scales---high-frequency, low-frequency, and macroscopic. 
\cite{Du14} studies statistical inference on compound Poisson processes under macroscopic observations. 
\cite{DuKa18} is another recent study on nonparametric estimation on compound Poisson processes under macroscopic observations. \cite{Co18b} discusses the robustness of spectral estimation of L\'evy measures of compound Poisson processes to $\Delta_n$, and it includes the consistency of the estimator under the macroscopic set up. 
Concerning recent contributions to nonparametric inference on L\'evy measures (or densities) under the high-frequency set up, we refer to \cite{Fi09a, Fi11a, Fi11}, \cite{Ve14}, \cite{KoPa16}, \cite{NiReSoTr16}, and \cite{KaKu17}. 
Recent studies on nonparametric estimation of L\'evy densities under the high-frequency scale are \cite{Sh06}, \cite{vaGuSp07}, \cite{CoGe09, CoGe10, CoGe11}, \cite{Fi09}, \cite{Gu09, Gu12}, \cite{NeRe09}, \cite{KaRe10}, \cite{Be11, Be11b}, \cite{Du13}, \cite{Ka14}, \cite{BeRe15}, and \cite{BeSc16}. 
Concerning literature on the low-frequency set up, we refer to \cite{NiRe12} for inference on L\'evy measures, and \cite{Pi94}, \cite{BuGr03}, and \cite{Co17} for nonparametric inference on compound Poisson processes. Further, \cite{ChDeHa10} and \cite{Tr15} investigate nonparametric estimation of a class of L\'evy processes under the low-frequency set up. 
\cite{BePaWo17} studies nonparametric estimation of L\'evy measures of the moving average L\'evy processes under low-frequency observations. 
\cite{BuVe13}, \cite{BuHoVeDe17}, and \cite{HoVe17} study nonparametic inference on L\'evy measures of It\^o semimartingales with L\'evy jumps under high-frequency observations. 
\cite{JoMeVa05} and \cite{IlMoRoSo15} investigate nonparametric estimation of the L\'evy-driven OU processes. 
\cite{JoMeVa05} derive consistency of their estimator for a class of L\'evy-driven OU processes, which include compound Poisson-driven OU processes. \cite{IlMoRoSo15} establish consistency of their estimator of the L\'evy density of (\ref{OU}) with compound Poisson subordinator in uniform norm at a polynomial rate. 
However, they do not derive limit distributions of their estimators.

The analysis of the present study is related to deconvolution problems for mixing sequence. \cite{Ma91, Ma93a, Ma93b} investigate the probability density deconvolution problems for $\alpha$-mixing sequences and derive convergence rates and asymptotic distributions of deconvolution estimators. Since the L\'evy-driven OU process (\ref{OU}) is $\beta$-mixing under some conditions (see \cite{Ma04} for details), our analysis can be interpreted as a deconvolution problem for a $\beta$-mixing sequence. However, we need a non-trivial analysis since we are considering additional structures emerging from the properties of the compound Poisson-driven OU process. To be more precise, \cite{Ma93b} assumes that, for a mixing sequence $\{\tilde{X}_{j}\}_{j \geq 0}$, the joint densities $p(x_{1},x_{j+1})$ of $\tilde{X}_{1}$ and $\tilde{X}_{j+1}$ are uniformly bounded for any $j \geq 1$ and $x_{1}, x_{j+1} \in \mathbb{R}$ to show the asymptotic independence of their estimators at different design points. Although we also observe a $\beta$-mixing sequence $\{X_{j\Delta}\}$ (see Remark 3.1 for details on the $\beta$-mixing property of $\{X_{j\Delta}\}$), we cannot assume such a condition directly in this study's context. Indeed, since the transition probability $P_{t}(x,dy)$ of $X$ has a point mass at $y = e^{-\lambda t}x$, $P_{t}(x, \cdot)$ does not have a transition density function (\cite{ZhShDe11}, Corollary 2). Therefore, to avoid such a problem, we consider the macroscopic regimes in this study.

The estimation problem of L\'evy measures is generally ill-posed in the sense of inverse problems, and the ill-posedness is induced by a decay of the characteristic function of a L\'evy process. We refer to \cite{NeRe09} as the seminal work in which such an explanation is given for the first time. In our case, the ill-posedness is induced by the decay of the characteristic function of the stationary distribution $\pi$ of the L\'evy-driven OU (\ref{OU}). In this sense, the problem in this study is a (nonlinear) inverse problem. \cite{Tr14} investigates conditions wherein a self-decomposable distribution is \textit{nearly} ordinary smooth, that is, the characteristic function of the self-decomposable distribution decays polynomially at infinity up to a logarithmic factor. \cite{Tr14b} applies those results to the nonparametric calibration of self-decomposable L\'evy option pricing models. Refining the result for a special case in \cite{Tr14}, we will show that the characteristic function of a self-decomposable distribution is regularly varying at infinity with some index $\alpha > 0$. This enables us to derive asymptotic distributions of the spectral estimator proposed in this study. 

Our analysis is also related to \cite{KaSa16} and \cite{KaKu17}. \cite{KaSa16} is a recent contribution to the literature on the construction of uniform confidence bands in probability density deconvolution problems for i.i.d. observations. The study formulates methods for constructing uniform confidence bands built on applications of intermediate Gaussian approximation theorems developed in \cite{ChChKa14a, ChChKa14b, ChChKa15, ChChKa16} and provides multiplier bootstrap methods for implementing uniform confidence bands. \cite{KaKu17} also develops confidence bands for L\'evy densities based on intermediate Gaussian and multiplier bootstrap approximation theorems. However, we adopt different methods for the construction of confidence bands. We derive high-dimensional central limit theorems based on intermediate Gaussian approximation for $\beta$-mixing process. Additionally, we can show that the variance-covariance matrix of the Gaussian random vector appearing in multivariate and high-dimensional central limit theorems is the identity matrix. Therefore, we do not need bootstrap methods to compute critical values of confidence bands.

The rest of the paper is organized as follows. In Section 2, we define a spectral estimator for the $k$-function. We give a multivariate central limit theorem of the spectral estimator in Section 3. In Section 4, we describe high-dimensional central limit theorems for the estimator and procedures for implementing confidence bands. In Section 5, we propose a practical method for bandwidth selection and report simulation results to study the finite sample performance of the spectral estimator. Discussions on our results and proposed confidence bands are presented in Section 6. All proofs are collated in Appendices A and B.

\subsection{Notation}
For any non-empty set $T$ and any (complex-valued) function $f$ on $T$, let $\|f\|_{T} = \sup_{t \in T}|f(t)|$, and, for $T = \mathbb{R}$, let $\|f\|_{L^{p}} = (\int_{\mathbb{R}}|f(x)|^{p}dx)^{1/p}$ for $p >0$. For any positive sequence $a_{n}, b_{n}$, we write $a_{n} \lesssim b_{n}$ if there is a constant $C >0$ independent of $n$ such that $a_{n} \leq Cb_{n}$ for all $n$,  $a_{n} \sim b_{n}$ if $a_{n} \lesssim b_{n}$ and $b_{n} \lesssim a_{n}$, and $a_{n} \ll b_{n}$ if $a_{n}/b_{n} \to 0$ as $n \to \infty$. For $a,b \in \mathbb{R}$, let $a \vee b = \max(a,b)$. For $a \in \mathbb{R}$ and $b>0$, we use the shorthand notation $[a \pm b] = [a-b, a+b]$. The transpose of a vector $x$ is denoted by $x^{\top}$. We use the notation $\stackrel{d}{\to}$ as convergence in the distribution. For random variables $X$ and $Y$, we write $X \stackrel{d}{=} Y$ if they have the same distribution. $N(\mu, \Sigma)$ denotes a (multivariate) normal distribution with a mean $\mu$ and a variance(-covariance matrix) $\Sigma$.

\section{Estimation of the $k$-function}
In this section, we introduce a spectral estimator for the L\'evy measure ($k$-function) of the L\'evy-driven OU process (\ref{OU}). First, we consider a symmetrized version of the $k$-function, that is, 
\[
k_{\sharp}(x) = 
\begin{cases}
k(x) & \text{if $x \geq 0$},\\
k(-x) & \text{if $x <0$},\\
\end{cases}
\]
A simple calculation yields 
\[
{1 \over \varphi(-t)} = \exp{\left(\int_{-\infty}^{0}(e^{itx}-1){k(-x) \over x}dx\right)}.
\]
Therefore, we have 
\begin{align*}
\varphi_{\sharp}(t) &:= {\varphi(t) \over \varphi(-t)} = \exp{\left(\int_{\mathbb{R}}(e^{itx}-1){k_{\sharp}(x) \over x}dx\right)}, \\
\varphi'_{\sharp}(t) &= {\varphi'(t)\varphi(-t) + \varphi(t)\varphi'(-t) \over \varphi^{2}(-t)} = {1 \over \varphi(-t)}\varphi'(t) -  \left({1 \over \varphi(-t)}\right)'\varphi(t) = i\left(\int_{\mathbb{R}}e^{itx}k_{\sharp}(x)dx\right)\varphi_{\sharp}(t). 
\end{align*}
This formally yields 
\[
k_{\sharp}(x) = {-i \over 2\pi}\int_{\mathbb{R}}e^{-itx}{\varphi'_{\sharp}(t) \over \varphi_{\sharp}(t)}dt. 
\]

Let
\begin{align*}
\widehat{\varphi}(u) &= {1 \over n}\sum_{j=1}^{n}e^{iuX_{j\Delta}},\ 
\widehat{\varphi}'_{\theta_{n}}(u) = {i \over n}\sum_{j=1}^{n}X_{j\Delta}e^{iuX_{j\Delta}}1\{|X_{j\Delta}| \leq \theta_{n}\}. 
\end{align*}
Here, $\theta_{n}$ is a sequence of constants such that $\theta_{n} \to \infty$ as $n \to \infty$ (in the rest of this study, we set $\theta_{n} \sim n^{1/2}(\log n)^{-3}$). 
Let $W:\mathbb{R} \to \mathbb{R}$ be an integrable (kernel) function such that $\int_{\mathbb{R}}W(x)dx = 1$, and its Fourier transform $\varphi_{W}$ is supported in $[-1,1]$ (i.e., $\varphi_{W}(u) = 0$ for all $|u|>1$). Then, the spectral estimator for $k$ at $x>0$ is defined by 
\[
\widehat{k}_{\sharp}(x) = {-i \over 2\pi}\int_{\mathbb{R}}e^{-itx}{\widehat{\varphi}'_{\sharp}(t) \over \widehat{\varphi}_{\sharp}(t)}\varphi_{W}(th)dt,
\]
where $h = h_{n}$ is a sequence of positive constants (bandwidths) such that $h_{n} \to 0$ as $n \to \infty$, and 
\begin{align*}
\widehat{\varphi}_{\sharp}(t) &= {\widehat{\varphi}(t) \over \widehat{\varphi}(-t)},\ \widehat{\varphi}'_{\sharp}(t) = {1 \over \widehat{\varphi}(-t)}\widehat{\varphi}'_{\theta_{n}}(t) + {\widehat{\varphi}'_{\theta_{n}}(-t) \over \widehat{\varphi}^{2}(-t)}\widehat{\varphi}(t). 
\end{align*}

In the following sections, we develop central limit theorems for $\widehat{k}$. 

\begin{remark}
We need the truncation in $\widehat{\varphi}'_{\theta_{n}}$ to show Lemma \ref{L1} in Appendix A by applying an exponential inequality for bounded mixing sequences. Additionally, refer to Remark 3.2 and the proof of Proposition 9.4 in \cite{Be10}. 
\end{remark}

\begin{remark}
For a complex value $a$, let $\overline{a}$ be the complex conjugate of $a$. We observe that $\widehat{k}_{\sharp}$ is real-valued. In fact, since $\overline{\widehat{\varphi}'_{\sharp}(t)} = -\widehat{\varphi}'_{\sharp}(-t)$ and $\overline{\widehat{\varphi}_{\sharp}(t)} = \widehat{\varphi}_{\sharp}(-t)$, by a change of variables, we have 
\begin{align*}
\overline{\widehat{k}_{\sharp}(x)} &= {i \over 2\pi}\int_{\mathbb{R}}e^{itx}{\overline{\widehat{\varphi}'_{\sharp}(t)} \over \overline{\widehat{\varphi}_{\sharp}(t)}}\overline{\varphi_{W}(th)}dt = {-i \over 2\pi}\int_{\mathbb{R}}e^{itx}{\widehat{\varphi}'_{\sharp}(-t) \over \widehat{\varphi}_{\sharp}(-t)}\varphi_{W}(-th)dt = \widehat{k}_{\sharp}(x). 
\end{align*}
Additionally, refer to Section \ref{Discuss} for detailed comments on the construction of the estimator $\hat{k}_{\sharp}$ and an alternative estimator. 
\end{remark}

\section{Multivariate Central Limit Theorem}\label{Main_results}
In this section, we present a multivariate central limit theorem for $\widehat{k}_{\sharp}$. 

\begin{assumption}\label{Ass1}
We assume the following conditions. 
\begin{itemize}
\item[(i)] $\int_{0}^{\infty}(1 \vee |x|^{2+\epsilon})k(x)dx<\infty$ for some $\epsilon>0$.
\item[(ii)] $k(0) = \nu((0,\infty)) = \alpha$ and $2<\alpha<\infty$. 
\item[(iii)] Let $r>1/2$, and let $p$ be the integer such that $p < r \leq p+1$. The function $k_{\sharp}$ is $p$-times differentiable, and $k^{(p)}_{\sharp}$ is $(r-p)$-H\"older continuous, that is, 
\[
\sup_{x,y \in \mathbb{R}, x \neq y}{|k^{(p)}_{\sharp}(x) - k^{(p)}_{\sharp}(y)| \over |x-y|^{r-p}} <\infty. 
\]
\item[(iv)] $|\varphi_{k}(u)| \lesssim (1 + |u|)^{-1}$ and $|\varphi'_{k}(u)| \vee |\varphi''_{k}(u)| \lesssim (1 + |u|)^{-2}$, where $\varphi_{k} (= \varphi'/(i\varphi))$ is the Fourier transform of $k$. 
\item[(v)] Let $W: \mathbb{R} \to \mathbb{R}$ be an integrable function such that 
\begin{align*}
\begin{cases}
\int_{\mathbb{R}}W(x)dx = 1,\ \int_{\mathbb{R}}|x|^{p+1}|W(x)|dx<\infty, \\
\int_{\mathbb{R}}x^{\ell}W(x)dx = 0,\ \ell = 1,\hdots, p,\\
\varphi_{W}(u) = 0,\ \forall |u| >1,\\
\text{$\varphi_{W}$ is three-times continuously differentiable,}
\end{cases}
\end{align*}
where $\varphi_{W}$ is the Fourier transform of $W$.
 
\item[(vi)] $\Delta = \Delta_{n} \geq {5C_{0} \over 4\beta_{1}(2 + 2\alpha - \delta)} \log n$, $n/\Delta \to \infty$, and 
\[
\left({(\log n)^{5} \over n}\right)^{1/(2+2\alpha-\delta)} \ll h \ll \left({1 \over n\log n}\right)^{1/(1+2r+2\alpha-\delta)}
\]
for some positive constant $C_{0}$ and $\delta \in (0,1/12)$ as $n \to \infty$. Here, $\beta_{1}$ is a positive constant. It appears in the mixing coefficient of $X = (X_{t})_{t \geq 0}$ (Conditions (i) and (ii) imply that $X$ is exponentially $\beta$-mixing with $\beta$-mixing coefficient $\beta_{X}(t) = O(e^{-\beta_{1}t})$ for some $\beta_{1}>0$. Refer to the following remark). 
\end{itemize}
\end{assumption}

\begin{remark}\label{Remark_Ass}
Conditions (i) and (ii) imply that the stationary distribution $\pi$ has a bounded continuous density (we also denote the density by $\pi$) such that $\|\pi\|_{\mathbb{R}} \lesssim 1$ and $\int_{\mathbb{R}}|x|\pi(dx)<\infty$ (see Lemma \ref{L0}). In this case, the stationary L\'evy-driven OU process defined by (\ref{OU}) is exponentially $\beta$-mixing (Theorem 4.3 in \cite{Ma04}), that is, the $\beta$-mixing coefficients for the stationary continuous-time Markov process $X$
\[
\beta_{X}(t) = \int_{\mathbb{R}}\|P_{t}(x,\cdot) - \pi(\cdot)\|_{TV}\pi(dx),\ t>0
\]
(this representation follows from Proposition 1 in \cite{Da73}) satisfy $\beta_{X}(t) = O(e^{-\beta_{1}t})$ for some $\beta_{1}>0$. Here, $P_{t}(x,\cdot)$ is the transition probability of the L\'evy-driven OU (\ref{OU}), and $\|\cdot\|_{TV}$ is the total variation norm. 

Condition (iii) is concerned with the smoothness of $k_{\sharp}$, and this condition is used to obtain a suitable bound of the deterministic bias of the estimator $\|[k_{\sharp}*(h^{-1}W(\cdot/h))] - k_{\sharp}\|_{\R}$. See Section \ref{Discuss} for details. 

Condition (iv) is satisfied if $k$ is two-times continuously differentiable on $(0,\infty)$ and $\int_{0}^{\infty}\{|k(x)| + |xk'(x)| + |x^{2}k''(x)|\}dx<\infty$. Indeed, by Condition (i), we have $|\varphi^{(p)}(u)| \lesssim 1$ for $p=0,1,2$. Additionally, by integration-by-parts and the Riemann-Lebesgue theorem, we also have that 
\begin{align*}
\left|\varphi_{k}(u)\right| &= \left|\int_{0}^{\infty}e^{iux}k(x)dx\right| = \left|{k(0+) \over iu} - {1 \over iu}\int_{0}^{\infty}e^{iux}k'(x)dx\right| \lesssim {1 \over |u|},\\
\left|\varphi'_{k}(u)\right| &= \left|{1 \over u}\varphi_{k}(u) + {1 \over iu^{2}}\int_{0}^{\infty}e^{iux}(k'(x) + xk''(x))dx\right| \lesssim {1 \over u^{2}},\\
|\varphi''_{k}(u)| &\leq  {2 \over u^{2}}|\varphi_{k}(u)| + {1 \over u^{2}}\left|\int_{0}^{\infty}e^{iux}(4xk'(x) + x^{2}k''(x))dx\right| \lesssim {1 \over u^{2}}
\end{align*}
as $|u| \to \infty$. 

Condition (v) is concerned with the kernel function $W$. We assume that $W$ is a $(p+1)$-th order kernel. However, we allow for the possibility that $\int_{\R} x^{p+1}W(x)dx=0$. 
It must be noted that since the Fourier transform of $W$ has compact support, the support of the kernel function $W$ is necessarily unbounded (see Theorem 4.1 in \cite{StWe71}).

Condition (vi) is concerned with the sampling frequency, bandwidth, and the sample size. The condition $\Delta \gtrsim \log n$ implies that we work with macroscopic observation scheme; this is a technical condition for the inference on $k$. We assume this condition to guarantee that the dependence among $\{X_{j\Delta}\}_{j=1}^{n}$ can be ignored asymptotically. We note that, to estimate $k$ uniformly on an interval $I \subset (0,\infty)$, we do not need the condition and we can work with the low-frequency set up (i.e., $\Delta>0$ is fixed). 
From a practical viewpoint, our methods could be applied to low-frequency data; additionally, it would work effectively if we suitably rescale the time scale of the data and if the sample size $n$ is sufficiently large. In our simulation study, we consider the case when $(n,\Delta) = (500,1)$, and our method functions effectively in this case. We also need Condition (vi) to derive the lower bound of $h$ for the uniform consistency of $\widehat{k}_{\sharp}(x)$ for $x = x_{\ell}$, $j=1,\hdots, N$ with $0<x_{1}<\cdots < x_{N}<\infty$. We need the upper bound of $h$ for the undersmoothing condition. Refer to Remark \ref{remark_band} of this study for comments on the condition on $h$. 

\end{remark}

To state a multivariate central limit theorem for $\widehat{k}_{\sharp}$, we introduce the notion of regularly varying functions. 
\begin{definition}[Regularly varying function]
A measurable function $U_{0}: [0,\infty) \to [0,\infty)$ is regularly varying at $\infty$ with index $\rho$ (written as $U_{0} \in RV_{\rho}$) if for $x>0$,
\[
\lim_{t \to \infty}{U_{0}(tx) \over U_{0}(t)} = x^{\rho}.
\]
\end{definition}
We say that a function $U$ is slowly varying if $U_{0} \in RV_{0}$. We refer to \cite{Re07} for details of regularly varying functions. The following lemma plays an important role in the proof of Theorem \ref{CLT2}.
\begin{lemma}\label{L-1}
Assume Condition (ii) in Assumption \ref{Ass1}. There exists a function $L : (1,\infty) \to [0,\infty)$, which slowly varies at $\infty$, and a constant $B>0$ such that 
\begin{align*}
\lim_{|t| \to \infty}{|t|^{\alpha}|\varphi(t)| \over L(|t|)} = B.
\end{align*}
\end{lemma}

\begin{remark}
In Assumption \ref{Ass1}, Condition (ii) is concerned with the smoothness of the stationary distribution $\pi$ of the L\'evy-driven OU process. Condition (ii) implies that the stationary distribution $\pi$ is \textit{nearly} ordinary smooth, that is, the characteristic function (\ref{ch_stmeas}) decays polynomially fast as $|u| \to \infty$ (Lemma \ref{L-1}), up to a slowly varying function. Since $k(x) = \nu((x,\infty))$, the finiteness of $k(0)$ is equivalent to the finiteness of the total mass of the L\'evy measure of the L\'evy process $J$. This means that the L\'evy process $J$ has finite activity, that is, it has only finitely many jumps in any bounded time interval. It is known that a L\'evy process with a finite L\'evy measure is a compound Poisson process. If $k(0) = \infty$, then the L\'evy process $J$ has infinite activity, that is, it has infinitely many jumps in any bounded time interval. In this case, the characteristic function (\ref{ch_stmeas}) decays faster than polynomials. Particularly, it decays exponentially fast as $|u| \to \infty$ if the Blumenthal-Getoor index of $J$ is positive, that is, if
\[
\rho_{BG} = \inf\left\{p>0: \int_{|x| \leq 1}|x|^{p}\nu(dx)<\infty\right\} >0. 
\] 
For example, this case includes inverse Gaussian, tempered stable, and normal inverse Gaussian processes. Condition (ii) rules out these examples since we could not construct confidence bands based on Gaussian approximation under our observation scheme (see the comments after Assumption 10 in \cite{KaSa16}). \cite{KaSa16} develops some methods to construct uniform confidence bands for the density deconvolution problem by using the intermediate Gaussian approximation. In their study, when the density of a measurement error is super smooth (this case corresponds to the case in our framework wherein the BG-index is positive), they assume that the effect of the estimation of the characteristic function of the measurement error based on $m = m_{n}$ auxiliary independent observations is asymptotically negligible, that is, $m_{n}/n \to \infty$ as $n \to \infty$. However, we can use $n$ observations to estimate $\varphi$ (this function corresponds to the characteristic function of a measurement error in deconvolution problems). Hence, in our situation, $m = n$. In this case, we can apply the results of the intermediate Gaussian approximation in \cite{ChChKa13} to the case wherein the density of a measurement error is ordinary smooth (or BG-index is $0$). However, to the best of our knowledge, such a result has not been achieved in the literature on deconvolution problems when the density of a measurement error is super smooth (or BG-index is positive). 
Therefore, we assume nearly ordinary smoothness of $\pi$ in our situation to obtain practical asymptotic theorems for the inference on $k$. 
\end{remark}

\begin{remark}\label{tail_decay}
Lemma \ref{L-1} implies that $|\varphi(u)|$ is a regularly varying function at $\infty$ with index $\alpha$. A slowly varying function $L(u)$ may go to $\infty$ as $u \to \infty$ but it does not grow faster than any power function, that is, 
\[
\lim_{u \to \infty}{L(u) \over u^{\delta}} = 0
\]
for any $\delta>0$. In fact, if $k(0) = \alpha>0$, from Proposition 1 in \cite{Tr14}, we have 
\[
(1+|u|)^{-\alpha} \lesssim |\varphi(u)| \lesssim (1 + |u|)^{-\alpha + \delta}.
\]
for any $\delta>0$. Such a tail behavior of $\varphi$ is related to Condition (vi) in Assumption \ref{Ass1}. If the stationary distribution $\pi$ is ordinary smooth, that is, $\varphi$ satisfies the relation
\[
(1 + |u|)^{-\alpha} \lesssim |\varphi(u)| \lesssim (1 + |u|)^{-\alpha}
\]
for some $\alpha>0$, then we can set $\delta = 0$ in Condition (vi). However, we must introduce $\delta>0$ to consider the effect of the slowly varying function $L$. 
\end{remark}

\begin{remark}\label{remark_band}
As shown in (\ref{Decomp_k}) and the comments below, if we do not assume the condition 
\[
h \ll \left({1 \over n\log n}\right)^{1/(1 + 2r + 2\alpha -\delta)},
\]
we have  
\[
\max_{1 \leq \ell \leq N}|\widehat{k}_{\sharp}(x_{\ell}) - k_{\sharp}(x_{\ell})| = O_{P}((nh^{2\alpha + 1-\delta})^{-1/2}\sqrt{\log n}) + O(h^{r})\ \text{as $n \to \infty$}
\]
where the second term of the right-hand side comes from the deterministic bias. For central limit theorems to hold and for constructing the confidence bands, we have to choose a bandwidth to ensure that the bias term is asymptotically negligible relative to the first term or ``variance'' term. The right-hand side is optimized if we take $h \sim ( \log n / n)^{1/(1 + 2r + 2\alpha-\delta)}$. 
\end{remark}

Under Assumption \ref{Ass1}, we can show that $\widehat{k}_{\sharp}(x) - k_{\sharp}(x)$ has the following asymptotically linear representation: 
\begin{align}\label{Approx}
\widehat{k}_{\sharp}(x) - k_{\sharp}(x) &= {-i \over 2\pi}\int_{\mathbb{R}}e^{-itx}\left({\widehat{\varphi}'_{\theta_{n}}(t) - \varphi'_{\theta_{n}}(t) \over \varphi(t)}\right)\varphi_{W}(th)dt + o_{P}((nh^{2\alpha + 1-\delta}\log n)^{-1/2}),
\end{align}
where $\varphi'_{\theta_{n}}(t) = E[\widehat{\varphi}'_{\theta_{n}}(t)]$. By a change of variables, we may rewrite the first term in (\ref{Approx}) as 
\begin{align}\label{EP approx}
Z_{n}(x) = {1 \over nh}\sum_{j=1}^{n}\left\{X_{j\Delta}1\{|X_{j\Delta}| \leq \theta_{n}\}K_{n}\left({x-X_{j\Delta} \over h}\right)- E\left[X_{1}1\{|X_{1}| \leq \theta_{n}\}K_{n}\left({x-X_{1} \over h}\right)\right]\right\},
\end{align}
where $K_{n}$ is a function defined by
\begin{align*}
K_{n}(x) &= {1 \over 2\pi}\int_{\mathbb{R}}e^{-itx}{\varphi_{W}(t) \over \varphi(t/h)}dt. 
\end{align*}

It must be noted that $K_{n}$ is well-defined and real-valued. 
To construct a confidence interval for $k(x)$, we estimate the variance of $\sqrt{n}hZ_{n}(x)$, which is $\sigma_{n}^{2}(x)$, by 
\begin{align}
\widehat{\sigma}_{n}^{2}(x) &= {1 \over n}\sum_{j=1}^{n}\left\{X_{j\Delta}1\{|X_{j\Delta}| \leq \theta_{n}\}\widehat{K}_{n}\left({x-X_{j\Delta} \over h}\right)\right\}^{2} \nonumber \\
&\quad - \left\{{1 \over n}\sum_{j=1}^{n}X_{j\Delta}1\{|X_{j\Delta}| \leq \theta_{n}\}\widehat{K}_{n}\left({x-X_{j\Delta} \over h}\right)\right\}^{2}, \label{sigsq_est}
\end{align}
where 
\[
\widehat{K}_{n}(x) = {1 \over 2\pi}\int_{\mathbb{R}}e^{-itx}{\varphi_{W}(t) \over \widehat{\varphi}(t/h)}dt. 
\]
\begin{remark}
We use Conditions (ii), (iv), and (v) in Assumption \ref{Ass1} to show that
\begin{align}\label{bound_Kn}
h^{\alpha}(|K_{n}(x)| + h|xK_{n}(x)|) \lesssim \min(1, 1/x^{2}).
\end{align}
Refer to the proof of Lemma \ref{L5} in Appendix A for details. Combining this bound on $K_{n}$ and Condition (vi) in Assumption \ref{Ass1}, we can show that the asymptotic variance-covariance matrix appearing in Theorem \ref{CLT2} is diagonal. 
\end{remark}
\begin{remark}
Propositions \ref{LemB1} and \ref{PrpB1} and Lemma \ref{L55} (see Appendix A) yield 
\begin{align*}
\sigma^{2}_{n}(x)  &= \Var(\sqrt{n}hZ_{n}(x)) \sim \Var(Z_{n,1}(x)) \gtrsim h^{-2\alpha+1-\delta}
\end{align*}
uniformly in $x \in I \subset (0,\infty)$ where $I$ is a compact set and $Z_{n,j}(x) = X_{j\Delta}1\{|X_{j\Delta}| \leq \theta_{n}\}K_{n}\left({x - X_{j\Delta} \over h}\right)$. Then, we can estimate $\sigma^{2}_{n}(x)$ by $\hat{\sigma}^{2}_{n}(x)$(see Lemma \ref{Var_approx} and the proof in Appendix A for details). 
\end{remark}

Now, we present the next multivariate central limit theorem. 
\begin{theorem}\label{CLT2}
Assume Assumption \ref{Ass1}. Then, for any $0<x_{1}<\hdots <x_{N}<\infty$, we have 
\[
\sqrt{n}h\left({\widehat{k}_{\sharp}(x_{1}) - k_{\sharp}(x_{1}) \over \widehat{\sigma}(x_{1})},\hdots, {\widehat{k}_{\sharp}(x_{N}) - k_{\sharp}(x_{N}) \over \widehat{\sigma}(x_{N})}\right)^{\top} \stackrel{d}{\to} N(0, I_{N}),
\] 
where $I_{N}$ is the $N$ by $N$ identity matrix and $\widehat{\sigma}_{n}(x) = \sqrt{\widehat{\sigma}^{2}_{n}(x)}$. 
\end{theorem}

\section{High-dimensional Central Limit Theorems}\label{HDCLT}
 In Section \ref{Main_results}, we present a multivariate (or finite-dimensional) central limit theorem for $\hat{k}_{\sharp}$. In this section, we present a high-dimensional central limit theorems as a refinement of Theorem \ref{CLT2}. Moreover, we propose some methods for constructing confidence bands for the $k$-function in Section \ref{Unif_CB} as an application of those results. 

\subsection{High-dimensional central limit theorems for $\widehat{k}_{\sharp}$}

For $1 \leq j \leq n$ and $1 \leq \ell \leq N$, let
\begin{align*}
Z_{n,j}(x_{\ell}) &= X_{j\Delta}1\{|X_{j\Delta}| \leq \theta_{n}\}K_{n}\left({x_{\ell} - X_{j\Delta} \over h}\right),\\
W_{n}(x_{\ell}) & = {1 \over \sigma_{n}(x_{\ell})\sqrt{n}}\sum_{j=1}^{n}(Z_{n,j}(x_{\ell}) - E[Z_{n,1}(x_{\ell})]) = {\sqrt{n}h \over \sigma_{n}(x_{\ell})}Z_{n}(x_{\ell}), 
\end{align*}
and let $I \subset (0,\infty)$ be an interval with finite Lebesgue measure $|I|$, $0<x_{1}<\cdots < x_{N}<\infty$, $x_{j} \in I$, $\ell=1,\hdots, N$. We assume that 
\begin{align}\label{design_points}
\min_{1 \leq k \neq \ell \leq N}|x_{k} - x_{\ell}| \gg h^{1-2\delta},
\end{align} 
and this implies that $N \ll h^{2\delta-1}$. Therefore, $N$ is allowed to go to infinity as $n \to \infty$. 

\begin{lemma}\label{Var_approx}
Under Assumption \ref{Ass1} and (\ref{design_points}), we have 
\begin{align*}
\max_{1 \leq \ell \leq N}\left|{\widehat{\sigma}_{n}^{2}(x_{\ell}) \over \sigma^{2}_{n}(x_{\ell})} -1\right| = o_{P}((\log n)^{-1}). 
\end{align*}
\end{lemma}
\begin{remark}
Since
\begin{align*}
\left|{\widehat{\sigma}_{n}^{2}(x) \over \sigma^{2}_{n}(x)} -1\right| = \left|{\widehat{\sigma}_{n}(x) \over \sigma_{n}(x)} -1\right|\left|{\widehat{\sigma}_{n}(x) \over \sigma_{n}(x)} + 1\right| \geq \left|{\widehat{\sigma}_{n}(x) \over \sigma_{n}(x)} -1\right|
\end{align*}
for any $0 < x < \infty$, Lemma \ref{Var_approx} implies 
\begin{align*}
\max_{1 \leq \ell \leq N}\left|{\widehat{\sigma}_{n}(x_{\ell}) \over \sigma_{n}(x_{\ell})} -1\right| = o_{P}((\log n)^{-1}). 
\end{align*}
\end{remark}
\begin{theorem}\label{Gauss_approx2}
Under Assumption \ref{Ass1} and (\ref{design_points}), we have 
\[
\sup_{t \in \mathbb{R}}\left|P\left(\max_{1 \leq \ell \leq N}|W_{n}(x_{\ell})| \leq t\right) - P\left(\max_{1 \leq \ell \leq N}|Y_{\ell}| \leq t\right)\right| \to 0,\ \text{as}\ n \to \infty, 
\]
where $Y = (Y_{1},\hdots, Y_{N})^{\top}$ is the standard normal random vector in $\mathbb{R}^{N}$. 
\end{theorem}

\begin{remark}\label{GA_proof}
Theorem \ref{Gauss_approx2} can be shown in two steps. In the first step, we approximate the distribution of $\max_{1 \leq \ell \leq N}|W_{n}(x_{\ell})|$ by that of $\max_{1 \leq \ell \leq N}|\check{Y}_{n, \ell}|$. Here, $\check{Y}_{n} = (\check{Y}_{n,1},\hdots, \check{Y}_{n,N})^{\top}$ is a centered normal random vector with covariance matrix $E[\check{Y}_{n}\check{Y}_{n}^{\top}] = q^{-1}E[W_{I_{1}}W_{I_{1}}^{\top}]$ where $q = q_{n}$ is a sequence of integers with $q_{n} \to \infty$ and $q_{n} = o(n)$ as $n \to \infty$, and 
\begin{align*}
W_{I_{1}} = \left(\sum_{k = 1}^{q}\left({Z_{n,k}(x_{1}) - E[Z_{n,1}(x_{1})] \over \sigma_{n}(x_{1})}\right),\hdots, \sum_{k = 1}^{q}\left({Z_{n,k}(x_{N}) - E[Z_{n,1}(x_{N})] \over \sigma_{n}(x_{N})}\right)\right)^{\top}.
\end{align*}
In the second step, we approximate the distribution of $\max_{1 \leq \ell \leq N}|\check{Y}_{n,\ell}|$ by that of $\max_{1 \leq \ell \leq N}|Y_{\ell}|$. For this, we compare the variance-covariance matrices $E[\check{Y}_{n}\check{Y}_{n}^{\top}]$ and $E[YY^{\top}] = I_{N}$ of two Gaussian random vectors $\check{Y}_{n}$ and $Y$ to establish 
\[
\sup_{t \in \mathbb{R}}\left|P\left(\max_{1 \leq \ell \leq N}|\check{Y}_{n}(x_{\ell})| \leq t\right) - P\left(\max_{1 \leq \ell \leq N}|Y_{\ell}| \leq t\right)\right| \to 0,\ \text{as}\ n \to \infty.
\]
Refer to proofs of Theorem \ref{Gauss_approx} and Proposition \ref{Gauss_compare} in Appendix A. 
\end{remark}

The well-known result in the extreme value theory shows that $\max_{1 \leq \ell \leq N}|Y_{\ell}| = O_{P}(\sqrt{\log N})$, for independent standard normal random variables $Y_{\ell}$, $\ell = 1,\hdots, N$ (see Example 1.1.7 in \cite{deFe06}). Then, Theorem \ref{Gauss_approx2} implies that $\max_{1 \leq \ell \leq N}|W_{n}(x_{\ell})| = O_{P}(\sqrt{\log n})$ since $\log N \lesssim \log (h^{2\delta-1}) \lesssim \log n$ under Assumption \ref{Ass1}. We can also show that 
\begin{align}\label{k_unif_approx}
{\sqrt{n}h(\hat{k}_{\sharp}(x_{\ell}) - k_{\sharp}(x_{\ell})) \over \sigma_{n}(x_{\ell})} &= W_{n}(x_{\ell}) + o_{P}((\log n)^{-1/2})
\end{align}
uniformly in $x \in \{x_{1},\hdots, x_{N}\}$. Therefore, together with Lemma \ref{Var_approx} and (\ref{k_unif_approx}), we have 
\begin{align*}
{\sqrt{n}h(\hat{k}_{\sharp}(x) - k_{\sharp}(x)) \over \hat{\sigma}_{n}(x)} &= {\sigma_{n}(x) \over \hat{\sigma}_{n}(x)} {\sqrt{n}h(\hat{k}_{\sharp}(x) - k_{\sharp}(x)) \over \sigma_{n}(x)}\\
&= {\sigma_{n}(x) \over \hat{\sigma}_{n}(x)}\{W_{n}(x) + o_{P}((\log n)^{-1/2})\}\ \text{(from (\ref{k_unif_approx}))}\\
&= \{1 + o_{P}((\log n)^{-1})\}\{W_{n}(x) + o_{P}((\log n)^{-1/2})\}\ \text{(from Lemma \ref{Var_approx})} \\
&= W_{n}(x) + o_{P}((\log n)^{-1/2})\ \text{(from $\max_{1 \leq \ell \leq N}|W_{n}(x_{\ell})| = O_{P}(\sqrt{\log n})$)}
\end{align*}
uniformly in $x \in  \{x_{1},\hdots,x_{N}\}$. 
This yields the following theorem.
\begin{theorem}\label{f_Gauss_approx}
Under Assumption \ref{Ass1} and (\ref{design_points}), we have 
\begin{align*}
\sup_{t \in \mathbb{R}}\left|P\left(\max_{1 \leq \ell \leq N}\left|{\sqrt{n}h(\widehat{k}_{\sharp}(x_{\ell}) - k_{\sharp}(x_{\ell})) \over \widehat{\sigma}_{n}(x_{\ell})}\right| \leq t\right) - P\left(\max_{1 \leq \ell \leq N}|Y_{\ell}| \leq t\right)\right| \to 0,\ \text{as}\ n \to \infty,
\end{align*}
where $Y = (Y_{1},\hdots, Y_{N})^{\top}$ is the standard normal random vector in $\mathbb{R}^{N}$. 
\end{theorem}

\subsection{Confidence bands for the $k$-function}\label{Unif_CB}
In this section, we discuss methods for constructing confidence bands for the $k$-function over $I = [a,b] \subset (0,\infty)$. Let $\xi_{1},\hdots, \xi_{N}$ be i.i.d. standard normal random variables, and, for $\tau \in (0,1)$, let $q_{\tau}$ satisfy 
\[
P\left(\max_{1 \leq j \leq N}|\xi_{j}| > q_{\tau}\right) = \tau.
\]
Then, 
\[
\widehat{C}_{1-\tau}(x_{\ell}) = \left[\widehat{k}_{\sharp}(x_{\ell}) \pm {\widehat{\sigma}_{n}(x_{\ell}) \over \sqrt{n}h}q_{\tau}\right],\ \ell = 1,\hdots, N 
\]
are joint asymptotic $100(1-\tau)$\% confidence intervals for $k_{\sharp}(x_{1}),\hdots, k_{\sharp}(x_{N})$. Theorem \ref{f_Gauss_approx} implies that we can construct confidence bands by linear interpolation of simultaneous confidence intervals $\{\hat{C}_{1-\tau}(x_{\ell})\}_{\ell = 1}^{N}$. 
If the sample size $n$ is sufficiently large, we can take a sufficiently large number of design points $N$. Therefore, proposed confidence bands can be arbitrary close to uniform confidence bands in such cases. We comment on the asymptotic validity of the confidence bands in Section \ref{Discuss}.

\section{Simulations}\label{Sec_sim}

\subsection{Simulation framework}
In this section, we present simulation results to see the finite-sample performance of the central limit theorems and the proposed confidence bands in Sections \ref{Main_results} and \ref{HDCLT}. We consider the following data generating process. 
\begin{align}\label{Levy_OU_DGP}
dX_{t} = -\lambda X_{t}dt + dJ_{\lambda t}
\end{align}
where $J_{t} = \sum_{j = 1}^{N_{t}}U_{j}$ is a compound Poisson process with intensity $\alpha$ and Gamma jump distribution with shape parameter $2$ and rate parameter $1$. Particularly, we consider three models, that is, $(\alpha, \lambda) = (2.1, 0.5), (3,0.5)$, and $(3,0.75)$.

As a kernel function, we use a flat-top kernel, which is defined by its Fourier transform
\begin{equation}
\varphi_{W}(u) =
\begin{cases}
1 &\text{if } \vert u \vert \leq c\\
\exp\left\{ \frac{-b \exp(-b/(\vert u \vert - c)^2)}{(\vert u \vert - 1)^2} \right\} &\text{if } c < \vert u \vert < 1\\
0 &\text{if } 1 \leq \vert u \vert 
\end{cases}
\label{eq: flap-top}
\end{equation}
where $0 < c < 1$ and $b > 0$. It must be noted that $\varphi_{W}$ is infinitely differentiable with $\varphi^{(\ell)}_{W}(0) = 0$ for all $\ell \geq 1$. This ensures that its inverse Fourier transform $W$ is of infinite order, that is, $\int_{\R} x^{\ell} W(x) dx=0$ for all integers $\ell \geq 1$ (cf. \cite{McPo04}). In our simulation study, we set $b = 1$ and $c = 0.05$. We also set the sample size $n$ and the time span $\Delta$ as $n = 500$ and $\Delta = 1$.

Now, we discuss bandwidth selection. We use a method that is similar to that proposed in \cite{KaKu17}. They adopt an idea of \cite{BiDuHoMu07} on bandwidth selection in density deconvolution. From a theoretical perspective, for our confidence bands to work, we have to choose bandwidths that are of a smaller order than the optimal rate for estimation under the loss function (or a ``discretized version'' of $L^{\infty}$-distance) $\max_{1 \leq \ell \leq N}| \hat{k}_{\sharp}(x_{\ell}) - k_{\sharp}(x_{\ell})|$. At the same time, choosing a very small bandwidth results in an extremely wide confidence band. Therefore, we should choose a bandwidth ``slightly'' smaller than the optimal one that minimizes $\max_{1 \leq \ell \leq N}| \hat{k}_{\sharp}(x_{\ell}) - k_{\sharp}(x_{\ell})|$. 
We employ the following rule for bandwidth selection. Let $\hat{k}_{h}$ be the spectral estimate with bandwidth $h$. 
\begin{enumerate}
\item Set a pilot bandwidth $h^{P} >0$ and make a list of candidate bandwidths $h_{j} = jh^{P}/J$ for $j=1,\dots,J$. 

\item Choose the smallest bandwidth $h_{j} \ (j \geq 2)$ such that the adjacent value $\max_{1 \leq \ell \leq N}| \hat k_{h_j}(x_{\ell}) - \hat k_{h_{j-1}}(x_{\ell})|$ is smaller than $\kappa \times \min \{ \max_{1 \leq \ell \leq N}| \hat{k}_{h_{k}}(x_{\ell}) - \hat{k}_{h_{k-1}}(x_{\ell}) | : k=2,\dots,J \}$ for some $\kappa > 1$. 
\end{enumerate}
In our simulation study, we set $h^{P} = 1, J = 20$, and $\kappa= 1.5$. This rule would choose a bandwidth ``slightly'' smaller than one 
that is intuitively the optimal bandwidth for the estimation of $k$ (as long as the threshold value $\kappa$ is reasonably chosen).

Figure \ref{fig:A2} shows five realizations of the discretized $L^{\infty}$-distance between the true $k$-function and estimates $\hat{k}_{\sharp}$ for different bandwidth values (left) and between the estimates of $k$ with adjacent bandwidth values (right) when $(\alpha, \lambda) = (2.1, 0.5)$. We find that the discretized $L^{\infty}$-distance between the estimates of $k$ with adjacent bandwidth values behave similarly to that between the true $k$-function and estimates $\hat{k}_{\sharp}$ for different bandwidth values. Hence, we can expect that, by using the proposed method for bandwidth selection, we can choose a ``good'' bandwidth for the construction of confidence bands.

\begin{remark}
In practice, it is also recommended to use visual information to find out on how $\max_{1\leq \ell \leq N}| \hat{k}_{h_{j}}(x_{\ell}) - \hat{k}_{h_{j-1}}(x_{\ell})|$ behaves as $j$ increases when determining the bandwidth. 
\end{remark}

Figure \ref{fig:A3} shows the normalized empirical distributions of $\widehat{k}_{\sharp}(x)$ at $x = 1.5$(left), $x = 2$(center), and $x = 2.5$(right) when $(\alpha, \lambda) = (2.1,0.5)$. The number of Monte Carlo iteration is 1,000 for each case. As seen from these figures, the central limit theorem implied by Theorem \ref{CLT2} holds true. 

Table \ref{table: T2_rev} presents simulation results of the cases when $(\alpha, \lambda) = (2.1, 0.5), (3,0.5)$, and $(3,0.75)$. 
We find that more accurate results are achieved when $\alpha = 3$ than when $\alpha = 2.1$. In general, the empirical coverage probabilities could be more accurate as the intensity of the Poisson process increases (see the comments on Figure \ref{fig:A4}). Overall, we can also find that the empirical coverage probabilities are reasonably close to the nominal coverage probabilities.

Figure \ref{fig:A4} shows the $85\%$(dark gray), $95\%$(gray), and $99\%$(light gray) confidence bands for the $k$-function when $(\alpha, \lambda) = (2.1,0.5)$. We find that the proposed confidence bands capture the monotonicity of the $k$-function and the width of confidence bands tend to increase as the design point becomes distant from the origin. The latter point can be partially attributed to the property of the L\'evy measure $\nu$ since the $k$-function is given by $k(x) = \nu((x,\infty))$ : For any (Borel) set $A \subset [0,\infty)$, $\nu(A)$ coincides with the expected number of jumps falling in $A$ in the unit time, that is, $\nu(A) = E[\sum_{0<t<1}1(J_{t} - J_{t-} \in A)]$, where $J_{t-} = \lim_{s \uparrow t}J_{s}$. Therefore, jumps of a larger size are less frequently observed since $\nu([0,\infty))<\infty$, in our simulation study. Further, the results also correspond to a well-known fact in nonparametric density estimation. Since few observations fall in the tail regions, the nonparametric estimation of a given density function tends to be less accurate in the tail area than in regions where the probability mass is concentrated. 

\begin{figure}[H]
  \begin{center}
    \begin{tabular}{cc}

      \begin{minipage}{0.5\hsize}
        \begin{center}
          \includegraphics[clip, width=6cm]{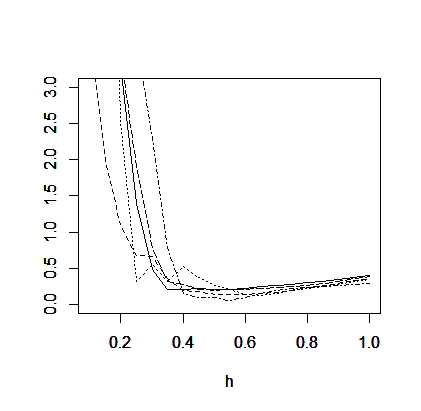}
        \end{center}
      \end{minipage}

      \begin{minipage}{0.5\hsize}
        \begin{center}
          \includegraphics[clip, width=6cm]{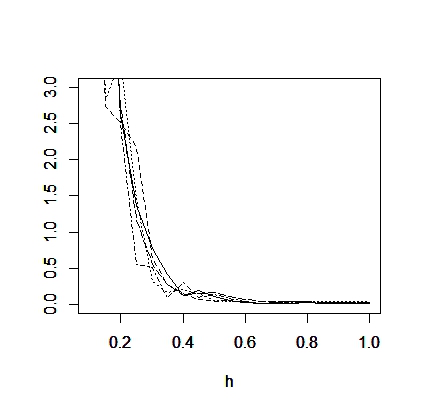}
        \end{center}
      \end{minipage}

    \end{tabular}
    \caption{Discrete $L^{\infty}$-distance between the true $k$-function and estimates $\hat{k}_{\sharp}$ (left) and between estimates of $k_{\sharp}$ (right) for different bandwidth values when $(\alpha, \lambda) = (2.1,0.5)$. We set $(n,\Delta) = (500, 1)$, $I = [1, 3]$, and $x_{\ell} = 1 + 0.2(\ell-1)$, $\ell = 1,\hdots, 11$. \label{fig:A2}}   
  \end{center}
\end{figure}

\begin{figure}[H]
  \begin{center}
    \begin{tabular}{ccc}

      \begin{minipage}{0.3\hsize}
        \begin{center}
          \includegraphics[clip, width=6cm]{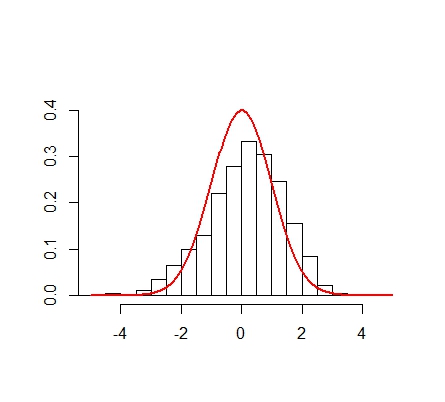}
        \end{center}
      \end{minipage}

      \begin{minipage}{0.3\hsize}
        \begin{center}
          \includegraphics[clip, width=6cm]{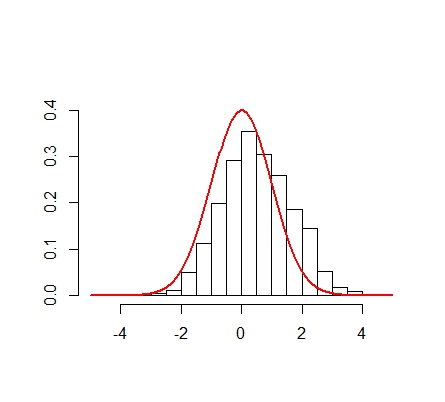}
        \end{center}
      \end{minipage}

      \begin{minipage}{0.33\hsize}
        \begin{center}
          \includegraphics[clip, width=6cm]{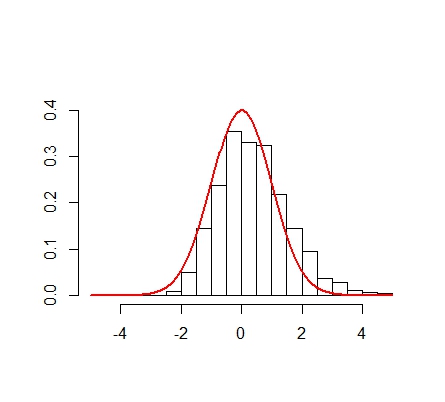}
        \end{center}
      \end{minipage}

    \end{tabular}
    \caption{Normalized empirical distributions of estimates at $x = 1.5$(left), $x=2$(center), and $x = 2.5$(right) when $(\alpha, \lambda) = (2.1,0.5)$. The red line is the density of the standard normal distribution. We set $(n, \Delta) = (500, 1)$.\label{fig:A3}}
  \end{center}
\end{figure}

{\small
\begin{table}[H]
\begin{center}
\begin{tabular}{cccccc}
\hline \hline
\multirow{2}{*}{\shortstack{Cov. Prob. \\$(1-\tau)$}}       &                          &         & \multicolumn{3}{c}{Model} \\ \cline{4-6}
                            &      & ($\alpha, \lambda$) & (2.1, 0.5) & ($3, 0.5$) & ($3, 0.75$)\\ \hline

\multirow{2}{*}{0.85}  & $I_{1}$  &   &  0.768       &  0.892        &0.848                 \\
                              & $I_{2}$   &  &   0.808      &  0.904        &0.888                 \\ \hline
\multirow{2}{*}{0.95}  & $I_{1}$  &  &  0.896       &0.976          &0.964           \\
                              & $I_{2}$   &  &   0.908      &0.972          &0.980           \\ \hline
\multirow{2}{*}{0.99}  & $I_{1}$  &  & 0.952        &0.988          &0.992        \\
                              & $I_{2}$  &   &0.956         &0.984          &0.996        \\ \hline \hline \\[5pt]
\end{tabular}
\caption{Empirical coverage probabilities of the confidence bands on $I_{1} = [1.5, 3.5]$ with $x_{\ell} = 1.5 + 0.2(\ell-1)$ and $I_{2} = [2,4]$ with $x_{\ell} = 2 + 0.2(\ell-1)$, $\ell = 1,\hdots,11$, based on 250 Monte Carlo repetitions. \label{table: T2_rev}}
\end{center}
\end{table}
}

\begin{figure}[H]
  \begin{center}
      \includegraphics[clip, width=6.5cm]{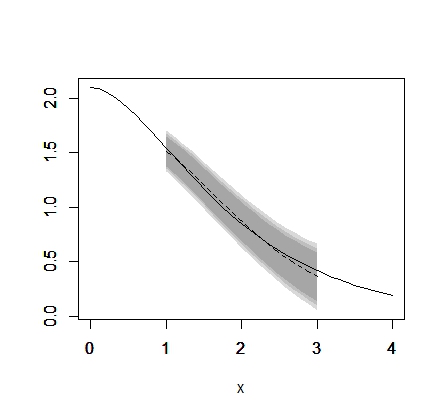}
     \caption{Estimates of $k$ with $85\%$(dark gray), $95\%$(gray), and $99\%$(light gray) confidence bands. The solid line corresponds to the true $k$-function. We set $(n,\Delta) = (500, 1)$, $I = [1,3]$, and $x_{\ell} = 1 + 0.2(\ell-1)$, $\ell = 1,\hdots, 11$. }
    \label{fig:A4}
  \end{center}
\end{figure}

\section{Discussions}\label{Discuss}

In this section, we discuss (1) the regularity condition on the $k$-function (Condition (iii) in Assumption \ref{Ass1}) and its relationship with the construction of our estimator, and (2) asymptotic properties of the proposed confidence bands. 

\subsection{Discussion on Condition (iii) in Assumption \ref{Ass1}}

We considered a symmetrized version of the $k$-function $k_{\sharp}$ and presented asymptotic properties of its estimator $\hat{k}_{\sharp}$. We also assumed a ``global'' regularity condition of $k_{\sharp}$ (Condition (iii) in Assumption \ref{Ass1}) to obtain a suitable bound of the deterministic bias of $\hat{k}_{\sharp}$. It must be noted that $k_{\sharp}$ is continuous at the origin, and if $k_{\sharp}$ has bounded $r$th derivative on $\mathbb{R}$ for some $r \geq 0$, then the deterministic bias of $\widehat{k}_{\sharp}$, which is given by $\|[k_{\sharp}*(h^{-1}W(\cdot/h))] - k_{\sharp}\|_{\R}$, is $O(h^{r})$ (Lemma \ref{L4} in Appendix A). However, if we restrict the class of kernel functions, which satisfy Condition (v) in Assumption \ref{Ass1}, then we can relax the ``global'' H\"older continuity. 

(i) When $1/2<r\leq 2$, we can use the symmetric second-order kernel functions. In this case, we can replace Condition (iii) in Assumption \ref{Ass1} with a ``local'' H\"older continuity of $k$ on $I^{\epsilon_{0}} = \{y \in \mathbb{R} :  |x - y| <\epsilon_{0}, \forall x \in I\}$, which does not include the origin. In fact, by taking a symmetric second-order kernel function $W_{2}$, we have, for any $x \in I$,
\begin{align*}
&\left | \int_{\R} \{ k_{\sharp}(x-yh) - k_{\sharp}(x)  \} W_{2}(y) dy \right |\\
&\quad = \left | \int_{|y| \leq \epsilon_{0}h^{-1}} \{ k(x-yh) - k(x)  \} W_{2}(y) dy \right |+ \left | \int_{|y|>\epsilon_{0}h^{-1}} \{ k(x-yh) - k(x)  \} W_{2}(y) dy \right |\\
&\quad \leq \left | \int_{|y| \leq \epsilon_{0}h^{-1}} \left [ \{ k(x-yh) - k(x) - \sum_{\ell=1}^{p} \frac{k^{(\ell)}(x)}{\ell !} (-yh)^{\ell}\} \right ] W_{2}(y) dy \right | +  2\|k\|_{\mathbb{R}}\int_{|y| > \epsilon_{0}h^{-1}}|W_{2}(y)|dy\\
&\quad \leq H_{0}h^{r} \int_{\R} |y|^{r} |W_{2}(y)| dy + {2h^{2}\|k\|_{\R} \over \epsilon_{0}^{2}}\int_{\mathbb{R}}|y|^{2}|W_{2}(y)|dy \lesssim h^{r}, 
\end{align*}
where $H_{0}:= \sup_{x,y \in I^{\epsilon_{0}}, x \neq y} \frac{|k^{(p)} (x) - k^{(p)}(y)|}{|x-y|^{r-p}} < \infty$, $\sum_{\ell=1}^{0} = 0$ and $0!=1$ by convention. We note that $k_{\sharp} = k$ on $I^{\epsilon_{0}}$. Hence, we can bound $\|[k_{\sharp}*(h^{-1}W_{2}(\cdot/h))] - k_{\sharp}\|_{I} = \|[k*(h^{-1}W_{2}(\cdot/h))] - k\|_{I}$. 

(ii) When $r> 2$, it would be difficult to weaken the global H\"older continuity assumption on $k_{\sharp}$ since symmetric ``finite order'' kernel functions do not satisfy higher-order properties. However, we can use the flat-top kernel function $W_{\infty}$, which is of ``infinite order,'' defined by its Fourier transform $\varphi_{W_{\infty}}$ to relax Condition (iii) in Assumption 3.1. Refer to (\ref{eq: flap-top}) for the definition. Indeed, $\varphi_{W_{\infty}}$ is infinitely differentiable and supported in $[-1,1]$; this implies that $|W_{\infty}(x)| = o(|x|^{-\ell})$ as $|x| \to \infty$ for all $\ell \geq 1$ (this follows from changes of variables) and $|x|^{r}|W(x)|$ is integrable. Then, we have 
\begin{align*}
&\left | \int_{\R} \{ k_{\sharp}(x-yh) - k_{\sharp}(x)  \} W_{\infty}(y) dy \right |\\
&\quad \leq \left | \int_{|y| \leq \epsilon_{0}h^{-1}} \left [ \{ k(x-yh) - k(x) - \sum_{\ell=1}^{p} \frac{k^{(\ell)}(x)}{\ell !} (-yh)^{\ell}\} \right ] W_{\infty}(y) dy \right | +  2\|k\|_{\mathbb{R}}\int_{|y| > \epsilon_{0}h^{-1}}|W_{\infty}(y)|dy\\
&\quad \leq H_{0}h^{r} \int_{\R} |y|^{r} |W_{\infty}(y)| dy + {2h^{r}\|k\|_{\R} \over \epsilon_{0}^{2}}\int_{\mathbb{R}}|y|^{r}|W_{\infty}(y)|dy \lesssim h^{r}.
\end{align*}
It is also shown that $\|[k*(h^{-1}W_{\infty}(\cdot/h))] - k\|_{I} \lesssim h^{r}$ for $1/2<r \leq 2$. 

Based on the discussion above, if we set the kernel function as the flat-top kernel $W_{\infty}$, then we can replace the global H\"older continuity (Condition (iii) in Assumption \ref{Ass1}) with the following local H\"older continuity. 

\textbf{Condition (iii)'} Let $r>1/2$, and let $p$ be the integer such that $p<r\leq p+1$. The function $k$ is $p$-times differentiable on $I^{\epsilon_{0}}$, which does not include the origin. Additionally, $k^{(p)}$ is $(r-p)$-H\"older continuous, that is, 
\[
H_{0}:=\sup_{x,y \in I^{\epsilon_{0}},x \neq y}{|k^{(p)}(x) - k^{(p)}(y)| \over |x-y|^{r-p}} <\infty. 
\]

Now, we set the kernel function $W = W_{\infty}$. In this case, we can use another natural (and simple) estimator for $k$ at $x>0$, which is given by
\[
\widehat{k}_{0}(x) = {-i \over 2\pi}\int_{\mathbb{R}}e^{-itx}{\widehat{\varphi}'_{\theta_{n}}(t) \over \widehat{\varphi}(t)}\varphi_{W_{\infty}}(th)dt.
\]
Additionally, Theorems \ref{CLT2} and \ref{f_Gauss_approx} hold by replacing $\hat{k}_{\sharp}$ with $\hat{k}_{0}$. We summarize the discussion so far as the following theorem. 

\begin{theorem}\label{non_symmetric_CLT}
Suppose Conditions (i), (ii), (iv), (v), and (vi) in Assumption \ref{Ass1}, and Condition (iii) hold true. Set the kernel function $W = W_{\infty}$. 
\begin{itemize}
\item[(i)] Then, for any $0<x_{1}<\hdots <x_{N}<\infty$, we have 
\[
\sqrt{n}h\left({\widehat{k}_{0}(x_{1}) - k(x_{1}) \over \widehat{\sigma}(x_{1})},\hdots, {\widehat{k}_{0}(x_{N}) - k(x_{N}) \over \widehat{\sigma}(x_{N})}\right)^{\top} \stackrel{d}{\to} N(0, I_{N}),
\] 
where $I_{N}$ is the $N$ by $N$ identity matrix and $\widehat{\sigma}_{n}(x) = \sqrt{\widehat{\sigma}^{2}_{n}(x)}$.
\item[(ii)] Additionally, suppose that (\ref{design_points}) holds. Then, we have 
\begin{align*}
\sup_{t \in \mathbb{R}}\left|P\left(\max_{1 \leq \ell \leq N}\left|{\sqrt{n}h(\widehat{k}_{0}(x_{\ell}) - k(x_{\ell})) \over \widehat{\sigma}_{n}(x_{\ell})}\right| \leq t\right) - P\left(\max_{1 \leq \ell \leq N}|Y_{\ell}| \leq t\right)\right| \to 0,\ \text{as}\ n \to \infty,
\end{align*}
where $Y = (Y_{1},\hdots, Y_{N})^{\top}$ is the standard normal random vector in $\mathbb{R}^{N}$. 
\end{itemize}
\end{theorem} 
We omit the proofs of Theorem \ref{non_symmetric_CLT} (i) and (ii) since the proofs are specializations of the proofs of Theorems \ref{CLT2} and \ref{f_Gauss_approx}.

\subsection{Discussion on the confidence bands}

Our method can be seen as an alternative method for constructing confidence bands based on a functional central limit theorem (FCLT) if the FCLT for the L\'evy measure $\nu$ is available (but to the best of our knowledge, such a result has not been achieved in the literature on nonparametric inference of L\'evy-driven SDEs). Moreover, the proofs clarify that if we strengthen the condition 
\[
h \ll \left({1 \over n\log n}\right)^{1/(1 + 2r + 2\alpha -\delta)}
\]
in Assumption 3.1 (vi) to $h^{r}\sqrt{nh^{2\alpha + 1 -\delta}(\log n)} = o(n^{-c})$ for some (sufficiently small) constant $c>0$, then there would exist a positive constant $c'$ such that the approximation of the high-dimensional central limit theorem holds at the rate $n^{-c'}$. This shows an advantage of our method to construct confidence bands based on the intermediate Gaussian approximation when compared to a method based on the Gumbel approximation. The coverage error of the latter is known to be logarithmically slow because of the slow convergence of normal extrema; refer to \cite{Ha91}. The proposed method is inspired by the idea developed in \cite{HoLe12}. If we take $x_{\ell} \in I$, $\ell = 1,\hdots, N$ to satisfy $\min_{1 \leq k \neq \ell \leq N}|x_{k} - x_{\ell}| = O(h^{1/2})$ (in this case, the condition (4.1) is satisfied), then $|x_{\ell} - x_{\ell-1}| \to 0$ uniformly for $\ell = 2,\hdots, N$. Therefore, for $x$ in $I$, 
\[
c_{L}(x) \leq k(x) \leq c_{U}(x)
\]
where 
\begin{align*}
c_{L}(x) &= \left({\widehat{k}_{\sharp}(x_{\ell}) - \widehat{k}_{\sharp}(x_{\ell-1}) - (\widehat{\sigma}_{n}(x_{\ell}) - \widehat{\sigma}_{n}(x_{\ell}))q_{\tau}/\sqrt{n}h \over x_{\ell}-x_{\ell-1}}\right)(x - x_{\ell-1}) +  \widehat{k}_{\sharp}(x_{\ell-1}) - {\widehat{\sigma}_{n}(x_{\ell-1}) \over \sqrt{n}h}q_{\tau},\\
c_{U}(x) &= \left({\widehat{k}_{\sharp}(x_{\ell}) - \widehat{k}_{\sharp}(x_{\ell-1}) + (\widehat{\sigma}_{n}(x_{\ell}) - \widehat{\sigma}_{n}(x_{\ell}))q_{\tau}/\sqrt{n}h \over x_{\ell}-x_{\ell-1}}\right)(x - x_{\ell-1}) +  \widehat{k}_{\sharp}(x_{\ell-1}) + {\widehat{\sigma}_{n}(x_{\ell-1}) \over \sqrt{n}h}q_{\tau}
\end{align*}
(if $x_{\ell-1} \leq x \leq x_{\ell}$ ($\ell = 2,\hdots,N$)) can be interpreted as an ``asymptotic'' $100(1-\tau)$\% uniform confidence band for $k$ on $I$. In fact, we can show that, as $n \to \infty$, 
\begin{align*}
P\left(\max_{1 \leq \ell \leq N}\left|{\sqrt{n}h(\widehat{k}_{\sharp}(x_{\ell}) - k_{\sharp}(x_{\ell})) \over \widehat{\sigma}_{n}(x_{\ell})}\right| \leq q_{\tau} \right) \to 1-\tau. 
\end{align*}
The same comments apply even if we replace $\hat{k}_{\sharp}$ with $\hat{k}_{0}$. See Appendix B for the asymptotic validity of the proposed confidence bands.

\section*{Acknowledgements}
I am grateful to the Editor Domenico Marinucci, an associate editor, and anonymous referees for their constructive comments that helped improve the quality of the paper. One referee kindly pointed out some relevant references that I had overlooked. I am also grateful to Kengo Kato for carefully reading the manuscript and for his helpful suggestions and encouragements. In addition, I thank Hiroki Masuda for his useful comments. This work is partially supported by Grant-in-Aid for Research Activity Start-up (19K20881) from the JSPS and the Research Institute of Mathematical Sciences, a Joint Usage/Research Center located in Kyoto University.

\appendix


\section{Proofs}

\subsection{Proofs for Section 3}\label{Proof_Sec3}
\begin{proof}[\textbf{Proof of Lemma \ref{L-1}}]
Observe that
\[
|\varphi(u)| = |\varphi(-u)| = \exp\left(\int_{0}^{\infty}(\cos(ux) -1){k(x) \over x}dx\right).
\]
For $x>1$, define 
\[
L(x) = \exp\left(\int_{1/x}^{1}(\alpha - k(y)){dy \over y}\right). 
\]
For any $\lambda >0$, 
\begin{align*}
{L(\lambda x) \over L(x)} &= \exp\left(\int_{1/\lambda x}^{1/x}(\alpha - k(y)){dy \over y}\right) = \exp\left(\int_{1/\lambda}^{1}(\alpha - k(z/x)){dz \over z}\right) \to 1,\ \text{as}\ x \to \infty. 
\end{align*}
Therefore, $L$ is a slowly varying function at $\infty$. Consider the following decomposition of $I(u):=\int_{0}^{\infty}(\cos (ux)-1)k(x)x^{-1}dx$.
\begin{align*}
I(u) &= \left(\int_{0}^{1/u} + \int_{1/u}^{1} + \int_{1}^{\infty}\right)(\cos(ux)-1){k(x) \over x}dx\\
&=: I_{1}(u) + I_{2}(u) + I_{3}(u). 
\end{align*} 
Now we evaluate three terms $I_{j}(u)$, $j=1,2,3$. 
First, by Riemann-Lebesgue theorem, 
\[
I_{3}(u) \to -\int_{1}^{\infty}{k(x) \over x}dx,\ \text{as}\ u \to \infty. 
\]
Moreover, 
\[
I_{1}(u) = \int_{0}^{1}(\cos(y)-1){k(y/u) \over y}dy \to \alpha \int_{0}^{1}(\cos(y)-1){dy \over y},\ \text{as}\ u \to \infty. 
\]
We also have that
\begin{align*}
I_{2}(u) + \alpha \log u - \log L(u) &= \int_{1/u}^{1}\cos(ux){k(x) \over x}dx\\
&= \int_{1}^{u}\cos(y){k(y/u) \over y}dy =: \widetilde{I}_{2}(u).
\end{align*}
Since $\int_{1}^{u}\cos(y)y^{-1}dy$ is convergent as $u \to \infty$ and $k$ is monotone decreasing function, we have that 
\[
\limsup_{u \to \infty}|\widetilde{I}_{2}(u)| \lesssim \left|\int_{1}^{\infty}{\cos(y) \over y}dy\right|<\infty.
\]
So, we complete the proof. 
\end{proof}

For the proof of Theorem \ref{CLT2}, we prepare some auxiliary results. 
\begin{lemma}\label{L0}
Assume Conditions (i), (ii) and (iv) in Assumption \ref{Ass1}. Then we have that
the measure $\pi$ and $x^{3}\pi(dx)$ has a bounded Lebesgue density on $\mathbb{R}$.
\end{lemma}
\begin{proof}
By Theorem 28.4 in \cite{Sa99}, $\pi$ has a bounded continuous Lebesgue density on $\mathbb{R}$. Also from the relation 
\begin{align*}
\varphi''(u) &= \varphi(u)\varphi_{k}^{2}(u) + \varphi(u)\varphi_{k}'(u),\\
\varphi'''(u) &= \varphi(u)\varphi_{k}^{3}(u) + 3\varphi(u)\varphi_{k}(u)\varphi'_{k}(u) + \varphi(u)\varphi''_{k}(u)\\
&= \left(\varphi(u)\varphi_{k}^{2}(u)\right)\varphi_{k}(u) + 3\left(\varphi(u)\varphi'_{k}(u)\right)\varphi_{k}(u) + \varphi(u)\varphi''_{k}(u),
\end{align*}
we see that
\begin{align*}
x^{2}\pi &= (k \ast \pi) \ast k + (xk) \ast \pi, \\
x^{3}\pi &= ((x^{2}\pi) - (xk) \ast \pi) \ast k + 3((xk) \ast \pi) \ast k + (x^{2}k) \ast \pi.
\end{align*} 
Therefore $x^{2}\pi$ has a Lebesgue density $x^{2}\pi(x)$ with 
\[
\|x^{2}\pi\|_{\mathbb{R}} \lesssim \|k\|_{\mathbb{R}}\|k\|_{L^{1}} + \|xk\|_{L^{1}} \lesssim 1.
\]
Here, $\|f\|_{L^{p}} = \left(\int_{\mathbb{R}}|f(x)|^{p}dx\right)^{1/p}$. Moreover, $x^{3}\pi$ has a Lebesgue density $x^{3}\pi(x)$ with 
\[
\|x^{3}\pi\|_{\mathbb{R}} \lesssim (\|x^{2}\pi\|_{\mathbb{R}} + \|xk\|_{L^{1}})\|k\|_{L^{1}} + 3\|xk\|_{L^{1}}\|k\|_{L^{1}} + \|x^{2}k\|_{L^{1}} \lesssim 1.
\]
\end{proof}

\begin{lemma}\label{L1}
Assume Conditions (i) and 
(vi) in Assumption \ref{Ass1}. Then we have  
\[
\|f_{1} - f_{2}\|_{[-h^{-1}, h^{-1}]} = O_{P}(n^{-1/2}\log n)
\]
for $(f_{1},f_{2}) = (\widehat{\varphi}, \varphi), (\widehat{\varphi}'_{\theta_{n}}, \varphi'_{\theta_{n}})$ where $\varphi'_{\theta_{n}}(u) := E[\widehat{\varphi}'_{\theta_{n}}(u)] $ and 
\begin{align*}
\|\varphi'_{\theta_{n}} - \varphi'\|_{[-h^{-1}, h^{-1}]} &= o(n^{-1/2}\log n),\\
\|\widehat{\varphi}'_{\theta_{n}} - \widehat{\varphi}'\|_{[-h^{-1}, h^{-1}]} &= o_{P}(n^{-1/2}\log n).
\end{align*}
\end{lemma}
\begin{proof}
The first result follows from Proposition 9.4 in \cite{Be11}. For the second result, 
we have that
\begin{align*}
|\varphi'_{\theta_{n}}(u) - \varphi'(u)| &\leq E\left[|X_{1}|1\{|X_{1}| > \theta_{n}\}\right]\\
&\leq E[|X_{1}|(|X_{1}|/\theta_{n})^{2}] \lesssim \theta_{n}^{-2} \ll n^{-1/2}\log n.
\end{align*}
We can also evaluate $\|\widehat{\varphi}'_{\theta_{n}} - \widehat{\varphi}'\|_{[-h^{-1}, h^{-1}]}$ in a similar way. 
\end{proof}

\begin{lemma}\label{L2}
Assume Condition (ii) in Assumption \ref{Ass1}. Then we have $\inf_{|u| \leq h^{-1}}|\varphi(u)| \gtrsim h^{\alpha}$. 
\end{lemma}
\begin{proof}
This result immediately follows from Remark \ref{tail_decay}.
\end{proof}

If we take $h$ sufficiently small, then Lemmas \ref{L1} and \ref{L2} imply that 
\[
\inf_{|u| \leq h^{-1}}|\widehat{\varphi}(u)| \geq \inf_{|u| \leq h^{-1}}|\varphi(u)| -o_{P}(h^{\alpha}) \gtrsim h^{\alpha} - o_{P}(h^{\alpha}), 
\]
so that with probability approaching one, $\inf_{|u|<h^{-1}}|\widehat{\varphi}(u)| \gtrsim h^{\alpha}$. 

\begin{lemma} \label{L3}
Assume Conditions (i), (iv) and (v) in Assumption \ref{Ass1}. Then we have that
\begin{align*}
&\left\|\left({\widehat{\varphi}'_{\sharp} \over \widehat{\varphi}_{\sharp}} - {\varphi'_{\sharp} \over \varphi_{\sharp}}\right) - {\widehat{\varphi}'_{\theta_{n}} - \varphi'_{\theta_{n}} \over \varphi}\right\|_{[-h^{-1},h^{-1}]}  = O_{P}(h^{-2\alpha}n^{-1}(\log n)^{2} + h^{1-\alpha}n^{-1/2}\log n). 
\end{align*}
\end{lemma}
\begin{proof}
(Step 1): First, we show that 
\[
\left\|\left({\widehat{\varphi}'_{\sharp} \over \widehat{\varphi}_{\sharp}} - {\varphi'_{\sharp} \over \varphi_{\sharp}}\right) - \left({1 \over \varphi_{\sharp}}\right)(\widehat{\varphi}'_{\sharp} - \varphi'_{\sharp})\right\|_{[-h^{-1},h^{-1}]} = O_{P}(h^{-2\alpha}n^{-1}(\log n)^{2} + h^{1-\alpha}n^{-1/2}\log n). 
\]
Consider the following decomposition.
\[
{\widehat{\varphi}'_{\sharp}(u) \over \widehat{\varphi}_{\sharp}(u)} - {\varphi'_{\sharp}(u) \over \varphi_{\sharp}(u)} = \left({1 \over \varphi_{\sharp}(u)}\right)'(\widehat{\varphi}_{\sharp}(u) - \varphi_{\sharp}(u)) + \left({1 \over \varphi_{\sharp}(u)}\right)(\widehat{\varphi}'_{\sharp} - \varphi'_{\sharp}) + R_{\sharp}(u),
\]
where 
\[
R_{\sharp}(u) = \left(1-{\widehat{\varphi}_{\sharp}(u) \over \varphi_{\sharp}(u)}\right)\left({\widehat{\varphi}'_{\sharp}(u) \over \widehat{\varphi}_{\sharp}(u)} - {\varphi'_{\sharp}(u) \over \varphi_{\sharp}(u)}\right). 
\]
We have that 
\[
\left\|\left({1 \over \varphi_{\sharp}}\right)'(\widehat{\varphi}_{\sharp} - \varphi_{\sharp})\right\|_{[-h^{-1},h^{-1}]} \lesssim \left\|\left({1 \over \varphi_{\sharp}}\right)'\right\|_{[-h^{-1},h^{-1}]}\|\widehat{\varphi}_{\sharp} - \varphi_{\sharp}\|_{[-h^{-1},h^{-1}]}
\] 
and 
\begin{align*}
&\|R_{\sharp}\|_{[-h^{-1}, h^{-1}]}\\
&\quad \lesssim \left\|{1 \over \varphi_{\sharp}}\right\|_{[-h^{-1}, h^{-1}]}\|\widehat{\varphi}_{\sharp} - \varphi_{\sharp}\|_{[-h^{-1}, h^{-1}]}\\
&\quad \quad \times \left(\left\|{1 \over \widehat{\varphi}_{\sharp}}\right\|_{[-h^{-1}, h^{-1}]}\|\widehat{\varphi}'_{\sharp} - \varphi'_{\sharp}\|_{[-h^{-1}, h^{-1}]} + \left\|{\varphi'_{\sharp} \over \widehat{\varphi}_{\sharp}\varphi_{\sharp}}\right\|_{[-h^{-1}, h^{-1}]}\|\widehat{\varphi}'_{\sharp} - \varphi'_{\sharp}\|_{[-h^{-1}, h^{-1}]}\right).
\end{align*}
In the rest of the proof, we write $\|\cdot\|_{[-h^{-1},h^{-1}]}$ as $\|\cdot\|$ for simplicity. Observe that 
\begin{align}\label{approx1}
\left\|{1 \over \varphi_{\sharp}}\right\| \lesssim 1,\ \left\|{1 \over \widehat{\varphi}_{\sharp}}\right\| = O_{P}(1)\ \text{and}\ \left\|\left({1 \over \varphi_{\sharp}}\right)'\right\| \lesssim h^{1-\alpha}.
\end{align} 
In fact, since we have that
\begin{align*}
\left\|{1 \over \widehat{\varphi}_{\sharp}}\right\| &\lesssim  \left\|{1 \over \widehat{\varphi}(-\cdot)}\right\|\|\widehat{\varphi} - \varphi\|  + \left\|{\varphi(\cdot) \over \widehat{\varphi}(-\cdot)} - {\varphi(\cdot) \over \varphi(-\cdot)}\right\| + \|\varphi_{\sharp}\|\\
&\lesssim \left\|{1 \over \widehat{\varphi}}\right\|\|\widehat{\varphi} - \varphi\| + \left\|{1 \over \widehat{\varphi}}\right\|\|\widehat{\varphi} - \varphi\| + \|\varphi_{\sharp}\| \lesssim O_{P}\left(h^{-\alpha}n^{-1/2}\log n\right) + 1 = O_{P}(1), 
\end{align*}
we obtain the second inequality. 
By Lemma \ref{L1}, we also have that
\begin{align}\label{approx2}
\|\widehat{\varphi}_{\sharp} - \varphi_{\sharp}\| &\lesssim \left\|{1 \over \widehat{\varphi}(-\cdot)}\right\|  \|\widehat{\varphi} - \varphi\|  + \left\|{1 \over \widehat{\varphi}(-\cdot)}\right\| \left\| \widehat{\varphi} - \varphi\right\| = O_{P}\left(h^{-\alpha}n^{-1/2}\log n\right).
\end{align}
Now we evaluate $\|\widehat{\varphi}'_{\sharp} - \varphi'_{\sharp}\|$. 
\begin{align*}
\|\widehat{\varphi}'_{\sharp} - \varphi'_{\sharp}\| \leq \|\widehat{\varphi}'_{\sharp} - \widetilde{\varphi}'_{\sharp}\| + \|\widetilde{\varphi}'_{\sharp} - \varphi'_{\sharp}\|,
\end{align*}
where
\[
\widetilde{\varphi}'_{\sharp}(t) = {\widehat{\varphi}'_{\theta_{n}}(t)\varphi(-t) + \varphi'(-t)\widehat{\varphi}(t) \over \varphi^{2}(-t)}.
\]
We observe that 
\begin{align*}
\|\widehat{\varphi}'_{\sharp} - \widetilde{\varphi}'_{\sharp}\| &\lesssim \left\|\widehat{\varphi}'_{\theta_{n}} \over \varphi^{2}(-\cdot)\right\| \times \left(\|\widehat{\varphi} - \varphi\| + \| \widehat{\varphi}'_{\theta_{n}} - \varphi'\| \right) \\
&\lesssim \left(\left\|\widehat{\varphi}'_{\theta_{n}}  - \varphi' \over \varphi^{2}(-\cdot)\right\| + \left\|{\varphi' \over \varphi^{2}(-\cdot)} \right\|\right)  \times \left(\|\widehat{\varphi} - \varphi\| + \| \widehat{\varphi}'_{\theta_{n}} - \varphi'\| \right) \\
&= O_{P}\left(h^{-2\alpha}n^{-1}(\log n)^{2} + h^{1-\alpha}n^{-1/2}\log n\right). \\
 \|\widetilde{\varphi}'_{\sharp} - \varphi'_{\sharp}\| &\lesssim \left\|{1 \over \varphi(-\cdot)}\right\| \times \|\widehat{\varphi}' - \varphi'\| + \left\|{\varphi'(-\cdot) \over \varphi^{2}(-\cdot)}\right\| \times \|\widehat{\varphi} - \varphi\| = O_{P}\left(h^{-\alpha}n^{-1/2}\log n\right).
\end{align*}
Then we have that 
\begin{align}\label{approx3}
\|\widehat{\varphi}'_{\sharp} - \varphi'_{\sharp}\| = O_{P}\left(h^{-\alpha}n^{-1/2}\log n\right). 
\end{align}
Together with (\ref{approx1}), (\ref{approx2}), and (\ref{approx3}), we have that
\[
\left\|\left({\widehat{\varphi}'_{\sharp} \over \widehat{\varphi}_{\sharp}} - {\varphi'_{\sharp} \over \varphi_{\sharp}}\right) - \left({1 \over \varphi_{\sharp}}\right)(\widehat{\varphi}'_{\sharp} - \varphi'_{\sharp})\right\| = O_{P}(h^{-2\alpha}n^{-1}(\log n)^{2} + h^{1-\alpha}n^{-1/2}\log n). 
\]

(Step 2): Next we show that 
\[
\left\| \left({1 \over \varphi_{\sharp}}\right)(\widehat{\varphi}'_{\sharp} - \varphi'_{\sharp}) - \left({1 \over \varphi}\right)(\widehat{\varphi}'_{\theta_{n}} - \varphi'_{\theta_{n}})\right\| = O_{P}(h^{1-\alpha}n^{-1/2}\log n). 
\] 
Observe that
\begin{align*}
&\left({1 \over \varphi_{\sharp}(u)}\right)(\widehat{\varphi}'_{\sharp}(u) - \varphi'_{\sharp}(u)) - \left({1 \over \varphi(u)}\right)(\widehat{\varphi}'_{\theta_{n}}(u) - \varphi'(u)) \\
& \quad = {\widehat{\varphi}'_{\theta_{n}}(u) \over \varphi(u)}\left({\varphi(-u) \over \hat{\varphi}(-u)}-1\right) + {\varphi(-u) \over \varphi(u)}\left({\widehat{\varphi}'_{\theta_{n}}(-u)\hat{\varphi}(u) \over \hat{\varphi}^{2}(-u)} - {\varphi'(-u)\varphi(u) \over \varphi^{2}(-u)}\right).
\end{align*}
Moreover, we have that 
\begin{align}\label{approx4}
\left\|{\widehat{\varphi}'_{\theta_{n}} \over \varphi}\left({\varphi(-\cdot) \over \hat{\varphi}(-\cdot)}-1\right)\right\| \lesssim \left\|{\varphi' \over \varphi^{2}}\right\| \times \|\hat{\varphi} - \varphi\| = O_{P}\left(h^{1-\alpha}n^{-1/2}\log n\right), 
\end{align}
and 
\begin{align}
\left\|{\varphi(-\cdot) \over \varphi}\left({\widehat{\varphi}'_{\theta_{n}}(-\cdot)\hat{\varphi} \over \hat{\varphi}^{2}(-\cdot)} - {\varphi'(-\cdot)\varphi \over \varphi^{2}(-\cdot)}\right)\right\| &\lesssim \left\|{1 \over \varphi}\right\| \times \left\| \widehat{\varphi}'_{\theta_{n}} - \varphi'\right\| + \left\|{\varphi' \over \varphi^{2}}\right\|\left\|\hat{\varphi} - \varphi\right\| \nonumber \\
&= O_{P}\left(h^{1-\alpha}n^{-1/2}\log n\right). \label{approx5}
\end{align}

Together with (\ref{approx4}) and (\ref{approx5}), we have that 
\begin{align}\label{approx6}
\left\| \left({1 \over \varphi_{\sharp}}\right)(\widehat{\varphi}'_{\sharp} - \varphi'_{\sharp}) - \left({1 \over \varphi}\right)(\widehat{\varphi}'_{\theta_{n}} - \varphi')\right\| = O_{P}(h^{1-\alpha}n^{-1/2}\log n). 
\end{align} 
Since $\|(\varphi'_{\theta_{n}} - \varphi')/\varphi\| \lesssim h^{-\alpha}\theta_{n}^{-2} \ll h^{1-\alpha}n^{-1/2}\log n$, we can replace $\varphi'$ with $\varphi'_{\theta_{n}}$ in (\ref{approx6}) and this completes the proof. 

\end{proof}

With almost the same arguments in the proof of Lemma \ref{L3}, we can show that 
\begin{align*}
&\left\|\left({\widehat{\varphi}'_{\theta_{n}} \over \widehat{\varphi}} - {\varphi' \over \varphi}\right) - {\widehat{\varphi}'_{\theta_{n}} - \varphi'_{\theta_{n}} \over \varphi}\right\|_{[-h^{-1},h^{-1}]}  = O_{P}(h^{-2\alpha}n^{-1}(\log n)^{2} + h^{1-\alpha}n^{-1/2}\log n). 
\end{align*}
Therefore, together with the result of Lemma \ref{L3}, we have that 
\begin{align*}
&\left\|\left({\widehat{\varphi}'_{\sharp} \over \widehat{\varphi}_{\sharp}} - {\varphi'_{\sharp} \over \varphi_{\sharp}}\right) - \left({\widehat{\varphi}'_{\theta_{n}} \over \widehat{\varphi}} - {\varphi' \over \varphi}\right)\right\|_{[-h^{-1},h^{-1}]}  = O_{P}(h^{-2\alpha}n^{-1}(\log n)^{2} + h^{1-\alpha}n^{-1/2}\log n). 
\end{align*}

\begin{lemma}\label{L5} 
We have that $h^{\alpha}(|K_{n}(x)| + h|xK_{n}(x)|) \lesssim \min(1, 1/x^{2})$.  
\end{lemma}
\begin{proof}
We first show $h^{\alpha}|K_{n}(x)| \lesssim \min(1, 1/x^{2})$. We follow the proof of Lemma 3 in \cite{Ma91}. By integration by parts, we have that 
\[
K_{n}(x) = {1 \over 2\pi x^{2}}\int_{\mathbb{R}}e^{-itx}\left({\varphi_{W}(t) \over \varphi(t/h)}\right)''dt.
\]
We also observe that  
\begin{align*}
\left({\varphi_{W}(t) \over \varphi(t/h)}\right)'' &= {\varphi''_{W}(t)\over \varphi(t/h)} -{2 \over h}{\varphi_{W}'(t)\varphi'(t/h) \over \varphi^{2}(t/h)} + {\varphi_{W}(t) \over h^{2}}\left(-{\varphi''(t/h) \over \varphi^{2}(t/h)} + 2{(\varphi'(t/h))^{2} \over \varphi^{3}(t/h)}\right)\\
&=: I_{1,n}(t) + I_{2,n}(t) + I_{3,n}(t). 
\end{align*}
Since $\varphi_{W}$ is supported in $[-1,1]$ and two-times differentiable, we can show 
\[
h^{\alpha}\int_{\mathbb{R}}|I_{j,n}(t)|dt\lesssim 1
\]
for $j=1,2,3$. Indeed, 
\begin{align*}
h^{\alpha}L(h^{-1})\int_{[-1,-{1 \over 2})\cup({1 \over 2},1]}|I_{1,n}(t)|dt &= \int_{[-1,-{1 \over 2})\cup({1 \over 2},1]}{|t|^{\alpha}|\varphi''_{W}(t)| \over |t/h|^{\alpha}L^{-1}(|t|/h)|\varphi(t/h)|}{L(1/h) \over L(|t|/h)}dt\\
& \lesssim \int_{\mathbb{R}}|t|^{\alpha}|\varphi''_{W}(t)|dt \lesssim 1, \\
h^{\alpha}L(h^{-1})\int_{[-1,-{1 \over 2})\cup({1 \over 2},1]}|I_{2,n}(t)|dt &= \int_{[-1,-{1 \over 2})\cup({1 \over 2},1]}{|t/h||\varphi_{k}(t/h)| \over |t/h|^{\alpha}L^{-1}(|t|/h)|\varphi(t/h)|}{L(1/h) \over L(|t|/h)}|t|^{\alpha-1}|\varphi'_{W}(t)|dt\\
& \lesssim \int_{\mathbb{R}}|t|^{\alpha-1}|\varphi'_{W}(t)|dt \lesssim 1,\\
h^{\alpha}L(h^{-1})\int_{[-1,-{1 \over 2})\cup({1 \over 2},1]}|I_{3,n}(t)|dt &\lesssim \int_{[-1,-{1 \over 2})\cup({1 \over 2},1]}\left({|t/h|^{2}|\varphi_{k}^{2}(t/h)| \over |t/h|^{\alpha}L^{-1}(|t|/h)|\varphi(t/h)|} + {|t/h|^{2}|\varphi'_{k}(t/h)| \over |t/h|^{\alpha}L^{-1}(|t|/h)|\varphi(t/h)|}\right)\\
&\quad \times {L(1/h) \over L(|t|/h)}|t|^{\alpha-2}|\varphi_{W}(t)|dt\\
& \lesssim \int_{\mathbb{R}}|t|^{\alpha-2}|\varphi_{W}(t)|dt \lesssim 1. 
\end{align*}
Moreover, we have that 
\begin{align*}
h^{\alpha}\int_{[-{1 \over 2},{1 \over 2}]}|I_{1,n}(t)|dt &= h^{\alpha}\int_{[-{1 \over 2},{1 \over 2}]}{|\varphi''_{W}(t)| \over|\varphi(t/h)|}dt \lesssim \int_{\mathbb{R}}(h+ |t|)^{\alpha}|\varphi''_{W}(t)|dt \lesssim 1, \\
h^{\alpha}\int_{[-{1 \over 2},{1 \over 2}]}|I_{2,n}(t)|dt &= h^{\alpha-1}\int_{[-{1 \over 2},{1 \over 2}]}{|\varphi_{k}(t/h)| \over |\varphi(t/h)|}|\varphi'_{W}(t)|dt \lesssim \int_{\mathbb{R}}{h^{\alpha}(1 + |t/h|)^{\alpha} \over h(1 + |t/h|)}|\varphi'_{W}(t)|dt\\
& \lesssim \int_{\mathbb{R}}(h + |t|)^{\alpha-1}|\varphi'_{W}(t)|dt \lesssim 1,\\
h^{\alpha}\int_{[-{1 \over 2},{1 \over 2}]}|I_{3,n}(t)|dt &\lesssim h^{\alpha-2}\int_{[-{1 \over 2},{1 \over 2}]}\left({|\varphi_{k}^{2}(t/h)| \over |\varphi(t/h)|} + {|\varphi'_{k}(t/h)| \over |\varphi(t/h)|}\right)|\varphi_{W}(t)|dt \lesssim \int_{\mathbb{R}}{h^{\alpha}(1 + |t/h|)^{\alpha} \over h^{2}(1 + |t/h|)^{2}}|\varphi'_{W}(t)|dt\\
& \lesssim \int_{\mathbb{R}}(h + |t|)^{\alpha-2}|\varphi_{W}(t)|dt \lesssim 1. 
\end{align*}
Since $\int_{\mathbb{R}}I_{j,n}(t)dt = \int_{[-1,1]}I_{j,n}(t)dt$ for $j=1,2,3$, we obtain the desired result.
Next we show $h^{\alpha +1}|xK_{n}(x)| \lesssim \min(1, 1/x^{2})$. Observe that 
\[
K_{n}(x) = {i \over 2\pi x^{3}}\int_{\mathbb{R}}e^{-itx}\left({\varphi_{W}(t) \over \varphi(t/h)}\right)'''dt
\]
and   
\begin{align*}
\left({\varphi_{W}(t) \over \varphi(t/h)}\right)''' &= {\varphi'''_{W}(t) \over \varphi(t/h)} -{3 \over h}{\varphi_{W}''(t)\varphi'(t/h) \over \varphi^{2}(t/h)} + 3{\varphi'_{W}(t) \over h^{2}}\left(-{\varphi''(t/h) \over \varphi^{2}(t/h)} + 2{(\varphi'(t/h))^{2} \over \varphi^{3}(t/h)}\right)\\
&\quad + {\varphi_{W}(t) \over h^{3}}\left(-{\varphi'''(t/h) \over \varphi^{2}(t/h)} + 6{\varphi'(t/h)\varphi''(t/h) \over \varphi(t/h)^{3}} - 6{(\varphi'(t/h))^{3} \over \varphi^{4}(t/h)}\right)\\
&=: \tilde{I}_{1,n}(t) + \tilde{I}_{2,n}(t) + \tilde{I}_{3,n}(t) + \tilde{I}_{4,n}(t). 
\end{align*}

We can show that $h^{\alpha+1}\int_{\mathbb{R}}|\tilde{I}_{j,n}(t)|dt \lesssim 1$, $j = 1,2,3$ and 
\begin{align*}
h^{\alpha+1}L(1/h)\int_{[-1,-{1 \over 2})\cup({1 \over 2},1]}\!\!\!|\tilde{I}_{4,n}(t)|dt &\lesssim \int_{[-1,-{1 \over 2})\cup({1 \over 2},1]}\left({h|t/h|^{3}|\varphi_{k}^{3}(t/h)| \over |t/h|^{\alpha}L^{-1}(|t|/h)|\varphi(t/h)|}\right. \\
&\left. \quad \quad  + \quad h{|t/h||\varphi_{k}(t/h)||t/h|^{2}|\varphi'_{k}(t/h)| \over |t/h|^{\alpha}L^{-1}(|t|/h)|\varphi(t/h)|} \right. \\
&\left. \quad \quad + \quad {|t/h|^{2}|\varphi''_{k}(t/h)| \over |t/h|^{\alpha}L^{-1}(|t|/h)|\varphi(t/h)|} \right){L(1/h) \over L(|t|/h)}|t|^{\alpha-3}|\varphi_{W}(t)|dt\\
&\lesssim \int_{\mathbb{R}}|t|^{\alpha-3}|\varphi_{W}(t)|dt \lesssim 1,
\end{align*}
\begin{align*}
h^{\alpha+1}\int_{[-{1 \over 2}, {1 \over 2}]}\!\!\!|\tilde{I}_{4,n}(t)|dt &\lesssim h^{\alpha-2}\int_{[-{1 \over 2}, {1 \over 2}]}\left({|\varphi_{k}^{3}(t/h)| \over |\varphi(t/h)|} + {|\varphi_{k}(t/h)||\varphi'_{k}(t/h)| \over |\varphi(t/h)|}  +  {|\varphi''_{k}(t/h)| \over |\varphi(t/h)|} \right)|\varphi_{W}(t)|dt\\
&\lesssim \int_{\mathbb{R}}(h + |t|)^{\alpha-3}|\varphi_{W}(t)|dt \lesssim 1
\end{align*} 
Therefore, we have the desired result. 
\end{proof}

Since
\[
h^{\alpha}yK_{n}\left({x-y \over h}\right) = -h^{\alpha+1}\left({x-y \over h}\right)K_{n}\left({x-y \over h}\right) + h^{\alpha}xK_{n}\left({x-y \over h}\right),
\]
Lemma \ref{L5} implies that each term on the right hand side is bounded (as a function of $y$) uniformly in $n$ and $x\in \{x_{1},\hdots, x_{N}\}$.

\begin{lemma}\label{L55}
Assume Conditions (i), (ii), (iv) and (v) in Assumption \ref{Ass1}. For any compact set $I$ such that $I \subset (0,\infty)$, we have that
\[
\int_{\mathbb{R}}K_{n}^{2}(x)dx \gtrsim h^{-2\alpha + \delta}.
\]
\end{lemma}
\begin{proof}
Let $\tilde{L}(x) = L(x)1\{x>1/2\} + 1\{0 \leq x \leq 1/2\}$. 
By Plancherel's theorem, we have that 
\begin{align*}
\int_{\mathbb{R}}K_{n}^{2}(x)dx = {1 \over 2\pi}\int_{\mathbb{R}}\left|{\varphi_{W}(t) \over \varphi(t/h)}\right|^{2}dt. 
\end{align*}
Now observe that
\begin{align*}
h^{2\alpha}\tilde{L}^{2}(1/h)\int_{\mathbb{R}}\left|{\varphi_{W}(t) \over \varphi(t/h)}\right|^{2}dt &= h^{2\alpha}\int_{[-{1 \over 2},{1 \over 2}]}\left|{\varphi_{W}(t) \over \varphi(t/h)}\right|^{2}dt\\
&\quad + \int_{[-1,-{1\over 2})\cup({1 \over 2},1]}{|t|^{2\alpha}|\varphi_{W}(t)|^{2} \over |(t/h)^{\alpha}L^{-1}(|t|/h)\varphi(t/h)|^{2}}{L^2(1/h) \over L^2(|t|/h)}dt.
\end{align*}
Since $|t|^{2\alpha}|\varphi_{W}(t)|^{2}$ is integrable and 
\begin{align*}
\lim_{h \to 0}{|t/h|^{\alpha}|\varphi(t/h)| \over L(|t|/h)}=: B,\ \lim_{h \to 0}{L(1/h) \over L(|t|/h)} = 1
\end{align*}
for any $|t|>0$, by dominated convergence theorem we have the desired result.
\end{proof}

\begin{lemma}\label{LemB2}
Let $Z_{n,j}(x) = X_{j\Delta}1\{|X_{j\Delta}| \leq \theta_{n}\}K_{n}((x-X_{j\Delta})/h)$. Then $\max_{1 \leq \ell \leq N}|E[Z_{n,1}(x_{\ell})]| \lesssim h$. 
\end{lemma}
\begin{proof}
Let $\widetilde{Z}_{n,j}(x) = X_{j\Delta}K_{n}((x-X_{j\Delta})/h)$.
By Fubini's theorem, we have that
\begin{align*}
\max_{1 \leq \ell \leq N}|E[\widetilde{Z}_{n,1}(x_{\ell})]|  &\leq h\max_{1 \leq \ell \leq N}\left|\int_{\mathbb{R}}k(x_{\ell}-hz)W(z)dz\right| \leq h\|k\|_{\mathbb{R}}\int_{\mathbb{R}}|W(y)|dy \lesssim h,\\
\max_{1 \leq \ell \leq N}E[|\widetilde{Z}_{n,1}(x_{\ell}) - Z_{n,1}(x_{\ell})|] &\leq \max_{1 \leq \ell \leq N}E[\widetilde{Z}_{n,1}^{2}(x_{\ell})]^{1/2}P(|X_{1}| > \theta_{n})^{1/2}\\
&\lesssim (h^{1-2\alpha})^{1/2} \times E[(|X_{1}|/\theta_{n})^{3}]^{1/2} = h^{1/2-\alpha}\theta_{n}^{-3/2} \lesssim h.
\end{align*}
Therefore, we have that 
\[
\max_{1 \leq \ell \leq N}|E[Z_{n,1}(x_{\ell})]| \leq \max_{1 \leq \ell \leq N}|E[\widetilde{Z}_{n,1}(x_{\ell})]| + \max_{1 \leq \ell \leq N}E[|\widetilde{Z}_{n,1}(x_{\ell}) - Z_{n,1}(x_{\ell})|] \lesssim h.
\]
\end{proof}

Lemmas \ref{L55} and \ref{LemB2} yield the following result on the lower bound of the variance of $Z_{n,1}(x)$. 
\begin{proposition}\label{LemB1} 
For any $\delta \in (0,1/12)$,
$\min_{1 \leq \ell \leq N}\Var(Z_{n,1}(x_{\ell})) \gtrsim h^{-2\alpha + \delta + 1}$.
\end{proposition}
\begin{proof}
Let $Z'_{n,j}(x) = X_{j\Delta}1\{|X_{j\Delta}| > \theta_{n}\}K_{n}((x-X_{j\Delta})/h)$. 
Observe that 
\[
\min_{1 \leq \ell \leq N}E[(Z'_{n,j})^{2}(x_{\ell})] \lesssim h^{-2\alpha}\theta_{n}^{-1}E[|X_{1}|^{3}] \lesssim h^{-2\alpha}\theta_{n}^{-1} \ll h^{-2\alpha + \delta + 1}. 
\]
Since $\min_{1 \leq \ell \leq N}E[\widetilde{Z}_{n,1}^{2}(x_{\ell})] \gtrsim h^{-2\alpha + \delta + 1}$ by Lemma \ref{L55}, we have that 
\[
\min_{1 \leq \ell \leq N}E[Z_{n,1}^{2}(x_{\ell})]  = \min_{1 \leq \ell \leq N}E[(\widetilde{Z}_{n,1}(x_{\ell}) - Z'_{n,1}(x_{\ell}))^{2}]  \sim \min_{1 \leq \ell \leq N}E[\widetilde{Z}_{n,1}^{2}(x_{\ell})]). 
\]
\end{proof}
 
\begin{lemma}\label{LemB3}
$\max_{1 \leq k,\ell \leq N}|\Cov(Z_{n,1}(x_{k}),Z_{n,j+1}(x_{\ell}))| \lesssim e^{-j\Delta\beta_{1}/3}h^{2/3-2\alpha}$.
\end{lemma}
\begin{proof}
Since $x^{3}\pi$ has a bounded Lebesgue density on $\mathbb{R}$ by Lemma \ref{L0} and $h^{2\alpha}|K_{n}|^{2}$ is integrable by Lemma \ref{L5}, we first observe that 
\begin{align*}
\max_{1 \leq \ell \leq N}E[|Z_{n,1}|^{3}(x_{\ell})] &\leq \max_{1 \leq \ell \leq N}\int_{\mathbb{R}}|y|^{3}\left|K_{n}\left({x_{\ell}-y \over h}\right)\right|^{3}\pi(y)dy \\
&\leq h\|y^{3}\pi\|_{\mathbb{R}}\|K_{n}^{3}\|_{L^{1}} \leq h\|y^{3}\pi\|_{\mathbb{R}}\|K_{n}\|_{\mathbb{R}}\|K_{n}^{2}\|_{L^{1}} \lesssim h^{1-3\alpha}. 
\end{align*}
Therefore, by Proposition 2.5 in \cite{FaYa03}, we obtain 
\begin{align*}
\max_{1 \leq k,\ell \leq N}|\Cov(Z_{n,1}(x_{k}), Z_{n,j+1}(x_{\ell}))| &\lesssim e^{-j\Delta \beta_{1}/3}\max_{1 \leq k \leq N}E[|Z_{n,1}(x_{k})|^{3}]^{1/3}\max_{1 \leq \ell \leq N}E[|Z_{n,j+1}(x_{\ell})|^{3}]^{1/3}\\ &\lesssim e^{-j\Delta \beta_{1}/3}h^{2/3-2\alpha}. 
\end{align*}
Then we have the desired result. 
\end{proof}

\begin{proposition}\label{PrpB1}
Let $\tilde{S}_{n}(x) = \sum_{j=1}^{n}Z_{n,j}(x)$. Then for any $\delta \in (0,1/12)$, we have that 
\begin{align*}
\max_{1 \leq \ell \leq N}\left({1 \over n}\Var(\tilde{S}_{n}(x_{\ell})) - \Var(Z_{n,1}(x_{\ell}))\right) = o(h^{-2\alpha + \delta +1}). 
\end{align*}
\end{proposition}
\begin{proof}
It is easy to show that 
\[
{1 \over n}\Var(\tilde{S}_{n}(x)) = \Var(Z_{n,1}(x)) + 2\sum_{j=1}^{n-1}(1-j/n)\Cov(Z_{n,1}(x), Z_{n,j+1}(x)).  
\]
By Lemma \ref{LemB3}, we have that 
\begin{align*}
h^{2\alpha -1-\delta}\max_{1 \leq \ell \leq N}\left|\sum_{j=1}^{\infty}\Cov(Z_{n,1}(x_{\ell}), Z_{n,j+1}(x_{\ell}))\right| &\leq h^{2\alpha-1-\delta}\max_{1 \leq \ell \leq N}\sum_{j=1}^{\infty}\left|\Cov(Z_{n,1}(x_{\ell}), Z_{n,j+1}(x_{\ell}))\right|\\
& \lesssim h^{2\alpha-1-\delta} \times h^{2/3-2\alpha}\sum_{j=1}^{\infty}e^{-j\Delta \beta_{1}/3}\\
& \lesssim h^{-\delta - 1/3}e^{-\Delta \beta_{1}/3} \lesssim e^{{5 \over 12}\log(1/h) - \Delta \beta_{1}/3}.
\end{align*}
Since $\log(1/h) < {C_{0} \over 2+2\alpha -\delta}\log n$ for sufficiently large $n$ and ${5C_{0} \over 4\beta_{1}(2 + 2\alpha -\delta)}\log n \leq \Delta$, we have that
\[
{5 \over 12}\log(1/h) - \Delta \beta_{1}/3 = -c_{0}\log n
\]
for some positive constant $c_{0}$. Therefore, we have the desired result. 
\end{proof}

Proposition \ref{PrpB1} implies that the dependence between $Z_{n,1}(x)$ and $Z_{n,j+1}(x)$ is negligible. This enables us to estimate $\sigma_{n}^{2}(x) = n^{-1}\Var(S_{n}(x)) = \Var(\sqrt{n}hZ_{n}(x))$ by the sample variance (\ref{sigsq_est}). 
Moreover Propositions \ref{LemB1} and \ref{PrpB1}, and Lemma \ref{L55} yield that $\min_{1 \leq \ell \leq N}\sigma^{2}_{n}(x_{\ell}) \gtrsim h^{-2\alpha + \delta + 1}$. 

Observe that 
\begin{align}
\widehat{k}_{\sharp}(x) - k_{\sharp}(x) &= {-i \over 2\pi}\int_{\mathbb{R}}e^{-iux}{\varphi'_{\sharp}(u) \over \varphi_{\sharp}(u)}\varphi_{W}(uh)du - k_{\sharp}(x) \nonumber  \\
&\quad + {-i \over 2\pi}\int_{\mathbb{R}}e^{-iux}\left({\widehat{\varphi}'_{\sharp}(u) \over \widehat{\varphi}_{\sharp}(u)} - {\varphi'_{\sharp}(u) \over \varphi_{\sharp}(u)}\right)\varphi_{W}(uh)du \nonumber  \\
&= [k_{\sharp} \ast (h^{-1}W(\cdot/h))](x) - k_{\sharp}(x) \nonumber  \\
&\quad +  {-i \over 2\pi}\int_{\mathbb{R}}e^{-iux}\left({\widehat{\varphi}'_{\sharp}(u) \over \widehat{\varphi}_{\sharp}(u)} - {\varphi'_{\sharp}(u) \over \varphi_{\sharp}(u)}\right)\varphi_{W}(uh)du \nonumber  \\
&=: I_{n} + II_{n}. \label{Decomp_k}
\end{align}
For the first term, we have that $\|I_{n}\|_{\mathbb{R}} \lesssim h^{r}$ (by Lemma \ref{L4}). For the second term $II_{n}$, Lemma \ref{L3} yields that 
\begin{align*}
II_{n} &= {-i \over 2\pi}\int_{\mathbb{R}}e^{-itx}\left({\widehat{\varphi}'_{\theta_{n}}(t) - \varphi'_{\theta_{n}}(t) \over \varphi(t)}\right)\varphi_{W}(th)dt + O_{P}(h^{-2\alpha-1}n^{-1}(\log n)^{2}  + h^{-\alpha}n^{-1/2}\log n)\\
&=  {-i \over 2\pi}\int_{\mathbb{R}}e^{-itx}\left({\widehat{\varphi}'_{\theta_{n}}(t) - \varphi'_{\theta_{n}}(t) \over \varphi(t)}\right)\varphi_{W}(th)dt + o_{P}((nh^{2\alpha + 1 -\delta}\log n)^{-1/2})
\end{align*}
uniformly in $x \in \{x_{1},\hdots, x_{N}\}$. Therefore, since $\min_{1 \leq \ell \leq N}\sigma_{n}(x_{\ell}) \gtrsim \sqrt{h^{-2\alpha + \delta +1}}$ (see the comment after Proposition \ref{PrpB1}), we have that
\begin{align}\label{k_approx}
{\sqrt{n}h(\hat{k}_{\sharp}(x) - k_{\sharp}(x)) \over \sigma_{n}(x)} &= W_{n}(x) + o_{P}((\log n)^{-1/2})
\end{align} 
uniformly in $x \in \{x_{1},\hdots, x_{N}\}$.

\begin{lemma}\label{L4}
Assume Conditions (iii), (v), and (vi) in Assumption \ref{Ass1}. Then we have that 
\[
\|[k_{\sharp} \ast (h^{-1}W(\cdot/h))] - k_{\sharp}\|_{\mathbb{R}} \lesssim h^{r} = o((nh^{2\alpha + 1 -\delta}\log n)^{-1/2}).
\] 
\end{lemma}
\begin{proof}
Observe that by a change of variables, $[k_{\sharp}*(h^{-1}W(\cdot/h))] (x) - k_{\sharp}(x) = \int_{\R} \{ k_{\sharp}(x-yh) - k_{\sharp}(x) \} W(y) dy$.
If $p \geq 1$, then
by Taylor's theorem, for any $x,y \in \R$, 
\[
k_{\sharp}(x-yh) - k_{\sharp}(x) 
=
 \sum_{\ell=1}^{p-1} \frac{k_{\sharp}^{(\ell)}(x)}{\ell !} (-yh)^{\ell} + \frac{k_{\sharp}^{(p)}(x-\theta yh)}{p!} (-yh)^{p}
\]
for some $\theta \in [0,1]$, where $\sum_{\ell=1}^{0} = 0$ by convention.
Since $k_{\sharp}^{(p)}$ is $(r-p)$-H\"{o}lder continuous, we have that $H:= \sup_{x,y \in \R, x \neq y} \frac{|k_{\sharp}^{(p)} (x) - k_{\sharp}^{(p)}(y)|}{|x-y|^{r-p}} < \infty$.
Now, since $\int_{\R} y^{\ell} W(y) dy = 0$ for $\ell=1,\dots,p$, we have that for any $x \in \R$, 
\begin{align*}
\left | \int_{\R} \{ k_{\sharp}(x-yh) - k_{\sharp}(x)  \} W(y) dy \right | &=\left | \int_{\R} \left [ \{ k_{\sharp}(x-yh) - k_{\sharp}(x)  \} - \sum_{\ell=1}^{p} \frac{k_{\sharp}^{(\ell)}(x)}{\ell !} (-yh)^{\ell} \right ] W(y) dy \right | \\
&\leq \frac{Hh^{r}}{p!} \int_{\R} |y|^{r} |W(y)| dy,
\end{align*}
where $0!=1$ by convention. 
This completes the proof. 
\end{proof}
Let $Q_{n}(x) = {1 \over \sqrt{n}}\sum_{j=1}^{n}Z_{n,j}(x)$ with $Z_{n,j}(x) = X_{j\Delta}1\{|X_{j\Delta}| \leq \theta_{n}\}K_{n}((x-X_{j\Delta})/h)$. We use the following result to show that the asymptotic variances which appear in Theorem
\ref{CLT2} is a diagonal matrix. 
\begin{proposition}\label{PrpB4}
For any $\delta \in (0,1/12)$, we have that 
\[
\max_{1 \leq k \neq \ell \leq N}|\Cov(Q_{n}(x_{k}),Q_{n}(x_{\ell}))| = o(h^{-2\alpha + \delta + 1}).
\] 
\end{proposition}
\begin{proof}
Since $\max_{1 \leq \ell \leq N}|E[Z_{n,1}(x_{\ell})]| \lesssim h$ by Lemma \ref{LemB2}, we have that 
\begin{align*}
\Cov(Q_{n}(x_{1}),Q_{n}(x_{2})) &= {1 \over n}\sum_{j,\ell=1}^{n}E[Z_{n,j}(x_{1})Z_{\ell,n}(x_{2})] - E[Z_{n,1}(x_{1})]E[Z_{n,1}(x_{1})]\\
&= E[Z_{n,1}(x_{1})Z_{n,1}(x_{2})] + 2\sum_{j=1}^{n-1}\left(1 - {j \over n}\right)E[Z_{n,1}(x_{1})Z_{n,j+1}(x_{2})] + O(h^{2}).
\end{align*}
With almost the same arguments in the proof of Proposition \ref{PrpB1} yields that 
\[
\max_{1 \leq k,\ell \leq N}\left(\sum_{j=1}^{n-1}E[|Z_{n,1}(x_{k})Z_{n,j+1}(x_{\ell})|]\right) = o(h^{-2\alpha + \delta + 1}). 
\]
Hence it is sufficient to show that $\max_{1 \leq k,\ell \leq N}E[|Z_{n,1}(x_{k})Z_{n,1}(x_{\ell})|] = o(h^{-2\alpha + \delta + 1})$. Let $0<x_{1}<x_{2}<\infty$. Since $h^{\alpha}|K_{n}(x)| \lesssim \min(1, 1/x^{2})$ by Lemma \ref{L5}, 
\begin{align*}
h^{2\alpha-1-\delta}E[|Z_{n,1}(x_{1})Z_{n,1}(x_{2})|] &= h^{2\alpha-1-\delta}\int_{\mathbb{R}}y^{2}\left|K_{n}\left({x_{1}-y \over h}\right)\right|\left|K_{n}\left( {x_{2}-y \over h}\right)\right|\pi(y)dy\\
&\leq h^{-\delta}\|x^{2}\pi\|_{\mathbb{R}}\int_{\mathbb{R}}|h^{\alpha}K_{n}(z)|\left|h^{\alpha}K_{n}\left(z + {x_{2} - x_{1} \over h}\right)\right|dz\\
&\lesssim h^{-\delta}\int_{\mathbb{R}}(1 \wedge z^{-2})\left(1\wedge {h^{2} \over (zh + (x_{2} - x_{1}))^{2}} \right)dz.
\end{align*}
If $|z|\leq h^{-2\delta}$ and take $h$ sufficiently small, then we have that 
\begin{align*}
\int_{|z| \leq h^{-2\delta}}(1 \wedge z^{-2})\left(1\wedge {h^{2} \over (zh + (x_{2} - x_{1}))^{2}} \right)dy &\leq \int_{|z| \leq h^{-2\delta}}(1 \wedge z^{-2}){h^{2} \over (x_{2} - x_{1})^{2}}dz\\
&\lesssim {h^{2} \over \min_{1 \leq k \neq \ell \leq N}|x_{k} - x_{\ell}|^{2}} \ll h^{4\delta}.  
\end{align*} 
Moreover, 
\begin{align*}
\int_{|z|>h^{-2\delta}}(1 \wedge z^{-2})\left(1\wedge {h^{2} \over (zh + (x_{2} - x_{1}))^{2}} \right)dy \leq \int_{|z|>h^{-2\delta}}(1 \wedge z^{-2})dz \lesssim h^{2\delta}. 
\end{align*}
Therefore we have that 
\[
h^{2\alpha-1-\delta}\max_{1 \leq k \neq \ell \leq N}E[|Z_{n,1}(x_{k})Z_{n,1}(x_{\ell})|] \lesssim h^{-\delta} (h^{4\delta} + h^{2\delta}) \lesssim h^{\delta} \ll 1. 
\]
\end{proof}

\begin{proof}[\textbf{Proof of Theorem \ref{CLT2}}]
Now we prove Theorem \ref{CLT2}. 
Let $S_{n}(x) = \sum_{j=1}^{n}Y_{n,j}(x)$ with $Y_{n,j}(x) = (Z_{n,j}(x) - E[Z_{n,1}(x)])$. First we will show that 
\[
{S_{n}(x) \over \sigma_{n}(x)\sqrt{n}} \stackrel{d}{\to} N(0,1)
\]
for $0<x<\infty$. We consider the following decomposition of $S_{n}(x)$.
\[
S_{n}(x) = \sum_{j=1}^{k_{n}}\xi_{n,j}(x) + \sum_{j=1}^{k_{n}}\eta_{n,j}(x) + \zeta_{n}(x),
\]
where
\begin{align*}
\xi_{n,j}(x) &= \sum_{k = (j-1)(l_{n} + s_{n})+1}^{jl_{n} + (j-1)s_{n}}Y_{n,k}(x),\ \eta_{n,j}(x) = \sum_{k = jl_{n} + (j-1)s_{n}+1}^{j(l_{n} + s_{n})}Y_{n,k}(x),\\
\zeta_{n}(x) &= \sum_{j=k_{n}(l_{n} + s_{n})}^{n}Y_{n,j}(x).
\end{align*}
We take $l_{n} = [\sqrt{nh}/(\log n)]$, $s_{n} = [(\sqrt{n/h} \log n)^{1/6}]$. Since $(\log n)^{4} \ll nh^{7/5}$,  we have that 
\[
{s_{n} \over l_{n}} = O\left(\left({1 \over nh^{7/5}}\right)^{5/12}(\log n)^{5/3}\right) \to 0
\] 
and $k_{n} = [n/(l_{n} + s_{n})] = O(\sqrt{n/h}\log n)$. 
We show the desired result in several steps. 

(Step1): In this step, we will show that 
\[
{S_{n}(x) \over \sigma_{n}(x)\sqrt{n}} = {1 \over \sigma_{n}(x)\sqrt{n}}\sum_{j=1}^{k_{n}}\xi_{n,j}(x) + o_{P}(1). 
\] 
Note that $\beta$-mixing coefficients satisfy $n^{6}\beta(n) \to 0$ as $n \to \infty$, we have that $k_{n}\beta(s_{n}) \to 0$ as $n \to \infty$. By the definition of $\eta_{n,1}(x)$, we have that 
\begin{align*}
{1 \over s_{n}\sigma_{n}^{2}(x)}\Var(\eta_{n,1}(x)) &\leq {\Var(Z_{n,1}(x)) \over \sigma_{n}^{2}(x)} + {1 \over \sigma_{n}^{2}(x)}\left|\sum_{j=1}^{s_{n}}\left(1-{j \over s_{n}}\right)\Cov(Z_{n,1}(x), Z_{n,j+1}(x))\right| \lesssim 1. 
\end{align*}
Since $|\eta_{n,j}(x)|/(s_{n}h^{-(1+\delta)/2}\sigma_{n}(x))$ is bounded (see the comment after the proof of Lemma \ref{L5}), by Proposition 2.6 in \cite{FaYa03}, $|\Cov(\eta_{n,1}(x), \eta_{n,j+1}(x))| \lesssim s_{n}^{2}h^{-(1+\delta)}\sigma_{n}^{2}(x)\beta(jl_{n}\Delta)$. Then we have that 
\[
{1 \over s_{n}\sigma_{n}^{2}(x)}\sum_{j=1}^{k_{n}}|\Cov(\eta_{n,1}(x), \eta_{n,j+1}(x))| \lesssim s_{n}h^{-(1+\delta)}\sum_{j=1}^{k_{n}}\beta(jl_{n}\Delta) \leq s_{n}h^{-(1+\delta)}\sum_{j=1}^{\infty}\beta(jl_{n}\Delta) \ll 1. 
\]
Therefore, we have that
\begin{align*}
{1 \over n \sigma_{n}^{2}(x)}\Var\left(\sum_{j=1}^{k_{n}}\eta_{n,j}(x)\right) &\lesssim {k_{n} \Var(\eta_{n,1}) \over n\sigma_{n}^{2}(x)} + {2 \over n\sigma_{n}^{2}(x)}\sum_{j=1}^{k_{n}-1}\left|\Cov(\eta_{n,1}(x), \eta_{n,j+1}(x))\right|\\
&\lesssim {k_{n}s_{n} \over n} + {2 \over n\sigma_{n}^{2}(x)}\sum_{j=1}^{k_{n}-1}\left|\Cov(\eta_{n,1}(x), \eta_{n,j+1}(x))\right| \to 0,\ \text{as}\ n \to \infty. 
\end{align*}
Likewise, we have that 
\[
{1 \over n\sigma_{n}^{2}(x)}\Var(\zeta_{n}(x)) = {l_{n} + s_{n} \over n}{1 \over (l_{n} + s_{n})\sigma_{n}^{2}(x)}\Var(\zeta_{n}(x)) \to 0,\ \text{as}\ n \to \infty
\]
since $n-k_{n}(l_{n} + s_{n}) \lesssim (l_{n} + s_{n})$.

(Step2): We set $T_{n}(x) = \sum_{j=1}^{k_{n}}\xi_{n,j}(x)$. In this step we show that
\[
{T_{n}(x) \over \sigma_{n}(x)\sqrt{n}} \stackrel{d}{\to} N(0,1). 
\]

Define $M_{n} = \left|E\left[\exp\left(itT_{n}(x)/\sqrt{n\sigma^{2}_{n}(x)}\right)\right] - \exp\left(-t^{2}/2\right)\right|$, where $i = \sqrt{-1}$. Then it is sufficient to show that for any $\epsilon>0$, $\lim_{n \to \infty}M_{n} < \epsilon$. Note that 
\begin{align*}
M_{n} &\leq \left|E\left[\exp(itT_{n}(x)/\sqrt{n\sigma^{2}_{n}(x)})\right] - \prod_{j=1}^{k_{n}}E\left[\exp(it\xi_{j,n}(x)/\sqrt{n\sigma^{2}_{n}(x)})\right]\right|\\
& + \left|\prod_{j=1}^{k_{n}}E\left[\exp(it\xi_{j,n}(x)/\sqrt{n\sigma^{2}_{n}(x)})\right] - \exp(-t^{2}/2)\right|\\
& =: A_{n,1} + A_{n,2}. 
\end{align*}
By Lemma 2.4 in \cite{FaMa92} and $k_{n}\beta(s_{n}) \to 0$ as $n \to \infty$, we have that $A_{n,1} \lesssim k_{n}\beta(s_{n}) \to 0$ as $n \to \infty$.  

Finally we show $\lim_{n \to \infty}A_{n,2} = 0$. This  is equivalent to showing that
\begin{align}\label{CLT_TL}
{1 \over \sqrt{n}}\widetilde{T}_{n}(x) \stackrel{d}{\to} N(0,1),
\end{align}
where $\widetilde{T}_{n}(x) = \sum_{j=1}^{n}\widetilde{\xi}_{n,j}$ and $\{\widetilde{\xi}_{n,j}(x)\}$ are independent random variables such that $\widetilde{\xi}_{n,j}(x) \stackrel{d}{=} \xi_{n,j}(x)/\sigma_{n}(x)$. It is easy to show that  $\{\xi_{n,j}(x)/\sigma_{n}(x)\}$ is a sequence of bounded random variables. To show (\ref{CLT_TL}), it is sufficient to check the following Lindeberg condition.
\[
{1 \over nh}\sum_{j=1}^{k_{n}}E[|\widetilde{\xi}_{n,j}(x)|^{2}1\{|\widetilde{\xi}_{n,j}(x)|> \omega \sqrt{n} \}] \to 0,\ \text{as}\ n \to \infty 
\]  
for any $\omega>0$. By H\"older's inequality, Markov's inequality and Proposition 2.7 in \cite{FaYa03}, we have that 
\begin{align*}
E[|\widetilde{\xi}_{n,j}|^{2}1\{|\widetilde{\xi}_{n,j}| \geq \omega \sqrt{n}\}] &\leq E[|\widetilde{\xi}_{n,j}|^{4}]^{1/2}P(|\widetilde{\xi}_{n,j}| > \omega \sqrt{n})^{1/2}\\
&\lesssim (l_{n}^{4/2})^{1/2}  {E[|\widetilde{\xi}_{n,j}|^{12}]^{1/2} \over n^{3}} \lesssim l_{n}  \left({l_{n} \over \sqrt{nh}}\right)^{3}  {1 \over (nh)^{3/2}}.
\end{align*}
Therefore, we have that
\[
{1 \over nh}\sum_{j=1}^{k_{n}}E[|\tilde{\xi}_{n,j}|^{2}1\{|\tilde{\xi}_{n,j}| > \omega \sqrt{n}\}] \lesssim {k_{n}l_{n} \over n}  \left({l_{n} \over \sqrt{nh}}\right)^{3}  \left({1 \over nh^{{5 \over 3}}}\right)^{2/3} \to 0,\ \text{as}\ n \to \infty 
\]
since $nh^{5/3} \to \infty$. 

(Step 3): In this step, we complete the proof. Considering (\ref{k_approx}), Condition (vi) in Assumption \ref{Ass1} and Lemma \ref{L4} yields that the bias term $I_{n}$ is asymptotically negligible since $h^{r}\sqrt{nh^{2\alpha + 1-\delta}\log n} \to 0$ as $n \to \infty$. This implies that 
\[
{\sqrt{n}h(\widehat{k}_{\sharp}(x) - k_{\sharp}(x)) \over \sigma_{n}(x)} - {S_{n}(x) \over \sigma_{n}(x)\sqrt{n}} = o_{P}((\log n)^{-1/2})
\]
and the asymptotic distribution of $\sqrt{n}(\widehat{k}_{\sharp}(x) - k_{\sharp}(x))$ is the same as that of $S_{n}(x)$. Moreover, Proposition \ref{PrpB4} implies that asymptotic covariance between $S_{n}(x_{1})/\sqrt{n}$ and $S_{n}(x_{2})/\sqrt{n}$ for different design points $0<x_{1}<x_{2}<\infty$ is asymptotically negligible. Therefore, we finally obtain the desired result. 
\end{proof}

\subsection{Proofs for Section \ref{HDCLT}}

We note that Lemmas and Propositions in Section \ref{Proof_Sec3} also hold when $0<x_{1}<\cdots<x_{N}<\infty$, $x_{\ell} \in I$ for $\ell = 1,\hdots, N$, and $\min_{1 \leq k \neq \ell \leq N}|x_{k} - x_{\ell}| \gg h^{1-2\delta}$.  In particular, we need to take into account the effect of  the separation between points in the proof of Lemmas 4.1 and A.10, and Theorem A.1. In the proof of Theorem A.1, we use the lower bound of $\min_{1 \leq \ell \leq N}\sigma_{n}(x_{\ell})$ to obtain an intermediate Gaussian approximation result.  We also need to take care of the effect of the discretization of a compact set $I$ to obtain the consistency of $\hat{\sigma}^{2}_{n}(x)$ on the discrete points in Lemma 4.1, that is, $\max_{1 \leq \ell \leq N}|\hat{\sigma}^{2}_{n}(x_{\ell})/\sigma^{2}_{n}(x_{\ell})-1| \stackrel{P}{\to} 0$. Moreover, in the proof of Lemma A.10, we use the condition $\min_{1 \leq k \neq \ell \leq N}|x_{k} - x_{\ell}| \gg h^{1-2\delta}$ to obtain a result that the variance-covariance matrix a random vector $(W_{n}(x_{1}),\hdots, W_{n}(x_{N}))^{\top}$ can be approximated by the $N \times N$ identity matrix and this yields a Gaussian comparison result (Proposition A.4). 

\begin{proof}[\textbf{Proof of Lemma \ref{Var_approx}}]
Since $\|K_{n}\|_{\mathbb{R}} \lesssim h^{-\alpha}$ and we can show $\|K_{n} - \widehat{K}_{n}\|_{\mathbb{R}} = O_{P}\left(h^{-2\alpha}n^{-1/2}\log n\right)$, we have that 
\[
\|\widehat{K}_{n}\|_{\mathbb{R}} \leq \|K_{n}\|_{\mathbb{R}} + \|K_{n} - \widehat{K}_{n}\|_{\mathbb{R}} \lesssim h^{-\alpha} + O_{P}\left(h^{-2\alpha}n^{-1/2}\log n\right) = O_{P}\left(h^{-\alpha}\right).
\]
Therefore, we have that $\|K_{n}^{2} - \widehat{K}_{n}^{2}\|_{\mathbb{R}} \leq \|K_{n} + \widehat{K}_{n}\|_{\mathbb{R}}\|K_{n} - \widehat{K}_{n}\|_{\mathbb{R}} = O_{P}\left(h^{-3\alpha}n^{-1/2}\log n\right)$. Then we have that 
\begin{align*}
&\max_{1 \leq \ell \leq N}\left|{1 \over n}\sum_{j=1}^{n}X_{j\Delta}1\{|X_{j\Delta}| \leq \theta_{n}\}\left\{\hat{K}_{n}((x_{\ell} - X_{j\Delta})/h) - K_{n}((x_{\ell} - X_{j\Delta})/h)\right\}\right|\\
&\quad \leq \underbrace{\left({1 \over n}\sum_{j=1}^{n}X_{j\Delta}1\{|X_{j\Delta}| \leq \theta_{n}\}\right)}_{= O_{P}(1)}\|\hat{K}_{n} - K_{n}\|_{\mathbb{R}} = O_{P}(h^{-2\alpha}n^{-1/2}\log n),
\end{align*}
and likewise, 
\begin{align*}
\max_{1 \leq \ell \leq N}\left|{1 \over n}\sum_{j=1}^{n}X_{j\Delta}^{2}1\{|X_{j\Delta}| \leq \theta_{n}\}\left\{\hat{K}_{n}^{2}((x_{\ell} - X_{j\Delta})/h) - K_{n}^{2}((x_{\ell} - X_{j\Delta})/h)\right\}\right| &= O_{P}(h^{-3\alpha}n^{-1/2}\log n). 
\end{align*}
Since $(h^{-2\alpha}n^{-1/2}\log n)^{2}/(h^{-3\alpha}n^{-1/2}\log n) = h^{-\alpha}n^{-1/2}\log n \ll 1$, we have that 
\[
\widehat{\sigma}_{n}^{2}(x) = \underbrace{{1 \over n}\sum_{j=1}^{n}Z_{n,j}^{2}(x) - \left({1 \over n}\sum_{j=1}^{n}Z_{n,j}(x)\right)^{2}}_{=: \tilde{\sigma}^{2}(x)} + O_{P}(h^{-3\alpha}n^{-1/2}\log n)
\]
uniformly $x = x_{\ell}$, $\ell =1,\hdots, N$. Furthermore, since $\min_{1 \leq \ell \leq N}\sigma^{2}_{n}(x_{\ell}) \gtrsim h^{-2\alpha + \delta +1}$ and  
\[
{h^{-3\alpha}n^{-1/2}\log n \over h^{-2\alpha + \delta +1}} = h^{-\alpha-\delta -1}n^{-1/2}\log n \ll (\log n)^{-1},
\]
it remains to prove that $\max_{1 \leq \ell \leq N}|\tilde{\sigma}^{2}(x_{\ell})/\sigma^{2}(x_{\ell}) -1 | = o_{P}((\log n)^{-1})$. Since $h^{\alpha}yK_{n}((x-y)/h)$ is uniformly bounded in $n$ and $x_{\ell}$ for $\ell = 1,\hdots, N$ (see the comment after the proof of Lemma \ref{L5}), we have that
\begin{align*}
\max_{1 \leq \ell \leq N}{E[|X_{1}|1\{|X_{1}|>\theta_{n}\}K_{n}((x_{\ell} - X_{1})/h)] \over h^{-\alpha + \delta/2 + 1/2}} &\lesssim {h^{-\alpha}P(|X_{1}|>\theta_{n}) \over h^{-\alpha + \delta/2 + 1/2}} \lesssim h^{-\delta/2 -1/2}\theta_{n}^{-2} \ll (\log n)^{-1/2},\\
\max_{1 \leq \ell \leq N}{E[X_{1}^{2}1\{|X_{1}|>\theta_{n}\}K_{n}^{2}((x_{\ell} - X_{1})/h)] \over h^{-2\alpha + \delta + 1}} &\lesssim {h^{-2\alpha}E[|X_{1}|1\{|X_{1}|>\theta_{n}\}] \over h^{-2\alpha + \delta + 1}}\\
& \lesssim h^{-\delta -1}E[|X_{1}|^{3}/\theta_{n}^{2}] \lesssim h^{-\delta -1}\theta_{n}^{-2} \ll (\log n)^{-1}.
\end{align*} 
Therefore, to complete the proof, it suffices to prove that 
\begin{align}
\max_{1 \leq \ell \leq N}\left|{1 \over n}\sum_{j=1}^{n}\left({Z_{n,j}^{2}(x_{\ell}) - E[Z_{n,j}^{2}(x_{\ell})] \over \sigma^{2}_{n}(x_{\ell})}\right)\right| &= o_{P}((\log n)^{-1}),\ \text{and} \label{sig_max1}\\
\max_{1 \leq \ell \leq N}\left|{1 \over n}\sum_{j=1}^{n}\left({Z_{n,j}(x_{\ell}) - E[Z_{n,j}(x_{\ell})] \over \sigma_{n}(x_{\ell})}\right)\right| &= o_{P}((\log n)^{-1/2}) \label{sig_max2}.
\end{align}
To prove (\ref{sig_max1}), we use Theorem 2.18 in \cite{FaYa03} with $b = h^{-\delta -1}$, $q = [h^{-\delta-2}] \wedge [n/2] \ll n$, and $\epsilon = \epsilon_{0}(\log n)^{-1}$ for any $\epsilon_{0}>0$ in their notations. Here, $[a]$ is the integer part of $a \in \mathbb{R}$. In this case we have that 
\begin{align*}
&P\left(\max_{1 \leq \ell \leq N}\left|{1 \over n}\sum_{j=1}^{n}\left({Z_{n,j}^{2}(x_{\ell}) - E[Z_{n,j}^{2}(x_{\ell})] \over \sigma^{2}_{n}(x_{\ell})}\right)\right| > \epsilon_{0}(\log n)^{-1}\right)\\
&\quad \leq \sum_{\ell = 1}^{N}P\left(\left|{1 \over n}\sum_{j=1}^{n}\left({Z_{n,j}^{2}(x_{\ell}) - E[Z_{n,j}^{2}(x_{\ell})] \over \sigma^{2}_{n}(x_{\ell})}\right)\right| > \epsilon_{0}(\log n)^{-1}\right)\\
&\quad \lesssim h^{-1+2\delta}\left(\exp\left(-{h^{-1} \over 8(\log n)^{2}}\right) + \sqrt{1 + h^{-\delta -1}(\log n) \over \epsilon_{0}}h^{-\delta-2}e^{-\Delta\beta_{1}nh^{\delta+2}}\right)  \to 0
\end{align*}
as $n \to \infty$, and likewise, we can show (\ref{sig_max2}). Therefore, we complete the proof. 
\end{proof}

Let $q > r$ be positive integers such that
\begin{align*}
q + r \leq n/2,\ q = q_{n} \to \infty,\ q_{n}= o(n),\ r = r_{n} \to \infty,\text{and}\ r_{n}= o(q_{n})\ \text{as}\ n \to \infty,
\end{align*}
and $m = m_{n} = [n/(q + r)]$. Consider a partition $\{I_{j}\}_{j = 1}^{m} \cup \{J_{j}\}_{j=1}^{m+1}$ of $\{1,\hdots, n\}$ where $I_{j} = \{(j-1)(q + r) +1,\hdots, jq + (j-1)r\}$, $J_{j} = \{jq + (j-1)r + 1,\hdots, j(q + r)\}$ and $J_{m+1} = \{m(q+r),\hdots, n\}$. 
First we show the following result on Gaussian approximation.
\begin{theorem}\label{Gauss_approx}
Under Assumption \ref{Ass1}, we have that 
\begin{align*}
\sup_{t \in \mathbb{R}}\left|P\left(\max_{1 \leq \ell \leq N}|W_{n}(x_{\ell})| \leq t\right) - P\left(\max_{1 \leq \ell \leq N}|\check{Y}_{\ell,n}| \leq t\right)\right| \to 0,\ \text{as}\ n \to \infty,
\end{align*}
where, $\check{Y}_{n} = (\check{Y}_{n,1},\hdots, \check{Y}_{n,N})^{\top}$ is a centered normal random vector with covariance matrix $E[\check{Y}_{n}\check{Y}_{n}^{\top}] = (mq)^{-1}\sum_{j = 1}^{m}E\left[W_{I_{j}}W_{I_{j}}^{\top}\right] = q^{-1}E\left[W_{I_{1}}W_{I_{1}}^{\top}\right]$ where 
\begin{align*}
W_{I_{j}} &= \left(\sum_{k \in I_{j}}\left({Z_{n,k}(x_{1}) - E[Z_{n,1}(x_{1})] \over \sigma_{n}(x_{1})}\right), \hdots, \sum_{k \in I_{j}}\left({Z_{n,k}(x_{N}) - E[Z_{n,1}(x_{N})] \over \sigma_{n}(x_{N})}\right)\right)^{\top}\\
&=: (W_{I_{j}}(x_{1}),\hdots, W_{I_{j}}(x_{N}))^{\top}. 
\end{align*}
\end{theorem}

\begin{proof}
Since $h^{\alpha}yK_{n}((x-y)/h)$ is uniformly bounded in $n$ and $x = x_{\ell}$, $\ell = 1,\hdots, N$ as a function of $y$ (see the comment after the proof of Lemma \ref{L5}) and $\min_{1 \leq \ell \leq N}\sigma_{n}(x_{\ell}) \gtrsim \sqrt{h^{-2\alpha + \delta + 1}}$, we have that 
\begin{align*}
|(Z_{n,j}(x_{\ell}) - E[Z_{n,j}(x_{\ell})])/\sigma_{n}(x_{\ell})| \lesssim h^{-(\delta+1)/2} 
\end{align*}
and $h^{-1/2}(\log Nn)^{5/2} \ll n^{1/8}$. Therefore, if we take $q_{n} = O(n^{q'})$ and $r_{n} = O(n^{r'})$ with $0< r' < q'  < 3/8$, we have that $q_{n}h^{-(\delta+1)/2}(\log Nn)^{5/2} \lesssim n^{1/2 - (1/8 + q')}$, $(r_{n}/q_{n})(\log N)^{2} \lesssim n^{-(q' - r')/2}$ and $m_{n}\beta_{X}(r_{n}) \lesssim m_{n}e^{-\beta_{1}r_{n}} \lesssim n^{-(q'-r')/2}$. Moreover, define 
\begin{align*}
\overline{\sigma}^{2}(q) &:= \max_{1 \leq \ell \leq N}\max_{I}\Var\left({1 \over \sigma_{n}(x_{\ell})\sqrt{q}}\sum_{k \in I}(Z_{n,k}(x_{\ell}) - E[Z_{n,1}(x_{\ell})])\right),\\
\underline{\sigma}^{2}(q)&:= \min_{1 \leq \ell \leq N}\min_{I}\Var\left({1\over \sigma_{n}(x_{\ell})\sqrt{q}}\sum_{k \in I}(Z_{n,k}(x_{\ell}) - E[Z_{n,1}(x_{\ell})])\right), 
\end{align*}
where $\max_{I}$ and $\min_{I}$ are taken over all $I \subset \{1,\hdots, n\}$ of the form $I = \{j + 1,\hdots, j + q\}$. By the stationarity of $\{X_{j\Delta}\}_{j \geq 0}$ and Proposition \ref{PrpB1}, we have that 
\begin{align*}
\overline{\sigma}^{2}(q) & = \overline{\sigma}^{2} \sim  \max_{1 \leq \ell \leq N}\left(\Var(Z_{n,1}(x_{\ell})/\sigma_{n}(x_{\ell})\right),\\
\underline{\sigma}^{2}(q)& = \underline{\sigma}^{2} \sim  \min_{1 \leq \ell \leq N}\left(\Var(Z_{n,1}(x_{\ell})/\sigma_{n}(x_{\ell})\right).
\end{align*}
Then there exists constants $0<c_{1}, C_{1}<\infty$ such that $c_{1} \leq \underline{\sigma}^{2}(q) \leq \overline{\sigma}^{2}(r)\vee \overline{\sigma}^{2}(q) \leq C_{1}$. 
From the above arguments, the conditions of Theorem B.1 in \cite{ChChKa13} are satisfied. So, we have the desired result.
\end{proof}

Next we show that the distribution of $\max_{1 \leq \ell \leq N}|\check{Y}_{n,\ell}|$ can be approximated by that of $\max_{1 \leq \ell \leq N}|Y_{\ell}|$ where $Y = (Y_{1},\hdots, Y_{N})^{\top}$ is a normal random vector in $\mathbb{R}^{N}$. For this, we prepare two lemmas. 
\begin{lemma}\label{Cov_approx}
Under Assumption \ref{Ass1}, we have that 
\begin{align*}
\max_{1 \leq k,\ell \leq N}\left|{\Cov(W_{n}(x_{k}), W_{n}(x_{\ell}))} - 1_{\left\{x_{k} = x_{\ell}\right\}}\right| = O(h^{\delta}). 
\end{align*}
\end{lemma}
\begin{proof}
Since the covariance between $Z_{n,j}(x_{\ell})$ and $Z_{n,k}(x_{\ell})$ for $j \neq k$ is asymptotically negligible with respect to the variances of each term by the proof of Proposition \ref{PrpB4}, it is sufficient to prove
\begin{align*}
\max_{1 \leq k,\ell \leq N}\left|{\Cov(Z_{n,1}(x_{k}), Z_{n,1}(x_{\ell})) \over \sqrt{\sigma_{n}^{2}(x_{k})\sigma_{n}^{2}(x_{\ell})}} - 1_{\left\{x_{k} = x_{\ell}\right\}}\right| = O\left(h^{\delta}\right).
\end{align*}
Since $1/\min_{1 \leq \ell \leq N}\sigma^{2}_{n}(x) \lesssim h^{2\alpha -\delta -1}$, from the same argument of the proof of Proposition \ref{PrpB4}, we have that 
\begin{align*}
&\max_{1 \leq k,\ell \leq N}\left|{\Cov(Z_{n,1}(x_{k}), Z_{n,1}(x_{\ell})) \over \sqrt{\sigma_{n}^{2}(x_{k})\sigma_{n}^{2}(x_{\ell})}} - 1_{\left\{x_{k} = x_{\ell}\right\}}\right|\\
&\quad = \max_{1 \leq k \neq \ell \leq N}\left|{\Cov(Z_{n,1}(x_{k}), Z_{n,1}(x_{\ell})) \over \sqrt{\sigma_{n}^{2}(x_{k})\sigma_{n}^{2}(x_{\ell})}}\right| \lesssim h^{2\alpha -\delta -1}\max_{1 \leq k \neq \ell \leq N}\left|\Cov(Z_{n,1}(x_{k}), Z_{n,1}(x_{\ell})) \right|\\
&\quad \lesssim {h^{2-\delta} \over \min_{1 \leq j \neq k \leq N}(x_{k} - x_{\ell})^{2}} \vee h^{\delta} \lesssim h^{\delta}
\end{align*}
since $\min_{1 \leq k \neq \ell \leq N}(x_{k} - x_{\ell})^{2} \gg h^{2-4\delta}$. Then we have the desired result. 
\end{proof}

\begin{lemma}\label{Gauss_Cov_compare}
Under Assumption \ref{Ass1}, we have that
\begin{align*}
\max_{1 \leq k,\ell \leq N}\left|q^{-1}\Cov(W_{I_{1}}(x_{k}), W_{I_{1}}(x_{\ell})) - 1_{\{x_{k} = x_{\ell}\}}\right| = O(h^{\delta}). 
\end{align*} 
\end{lemma}
\begin{proof}
Form the same argument of the proof Propositions \ref{PrpB1} and \ref{PrpB4}, 
\[
{1 \over q}\sum_{k,\ell \in I_{j}, k \neq \ell}{\Cov(Z_{n,k}(x_{m_{1}}), Z_{n,\ell}(x_{m_{2}}))\over \sigma_{n}(x_{m_{1}})\sigma_{n}(x_{m_{2}})}
\]
is asymptotically ignorable for $1 \leq m_{1}, m_{2} \leq N$. Therefore, the proof of Lemma \ref{Cov_approx} yields that  
\begin{align*}
&\max_{1 \leq k,\ell \leq N}\left|q^{-1}\Cov(W_{I_{1}}(x_{k}), W_{I_{1}}(x_{\ell})) - 1_{\{x_{k} = x_{\ell}\}}\right|\\
&\quad = O\left(\max_{1 \leq k,\ell \leq N}\left|{\Cov(Z_{n,1}(x_{k}), Z_{n,1}(x_{\ell})) \over \sqrt{\sigma_{n}^{2}(x_{k})\sigma_{n}^{2}(x_{\ell})}} - 1_{\left\{x_{k} = x_{\ell}\right\}}\right|\right) = O(h^{\delta}).
\end{align*}
This completes the proof. 
\end{proof}

Lemma \ref{Gauss_Cov_compare} and Condition (vi) in Assumption \ref{Ass1} yield the following result on Gaussian comparison:
\begin{proposition}\label{Gauss_compare}
Under Assumption \ref{Ass1}, we have that 
\[
\sup_{t \in \mathbb{R}}\left|P\left(\max_{1 \leq \ell \leq N}|\check{Y}_{n,\ell}| \leq t\right) - P\left(\max_{1 \leq \ell \leq N}|Y_{\ell}| \leq t\right)\right| \to 0,\ \text{as}\ n \to \infty, 
\]
where $Y = (Y_{1},\hdots, Y_{N})^{\top}$ is a standard normal random vector in $\mathbb{R}^{N}$. 
\end{proposition}
\begin{proof}
Let $\Delta(\check{Y}_{n}, Y):= \max_{1 \leq k,\ell \leq N}\left|\Cov(\check{Y}_{n,k}, \check{Y}_{n,\ell}) - 1_{\{x_{k} = x_{\ell}\}}\right|$. By Lemma \ref{Gauss_Cov_compare} and Theorem 2 in \cite{ChChKa15}, we have that  
\begin{align*}
\sup_{t \in \mathbb{R}}\left|P\left(\max_{1 \leq \ell \leq N}|\check{Y}_{n,\ell}| \leq t\right) - P\left(\max_{1 \leq \ell \leq N}|Y_{\ell}| \leq t\right)\right| &\lesssim \Delta(\check{Y}_{n}, Y)^{1/3}\{1 \vee \log (N/\Delta(\check{Y}_{n}, Y))\}^{2/3} \to 0
\end{align*}
as $n \to \infty$. Therefore, we obtain the desired result.
\end{proof}

\begin{proof}[\textbf{Proof of Theorem \ref{Gauss_approx2}}]
Theorem \ref{Gauss_approx2} immediately follows from Theorem \ref{Gauss_approx} and Proposition \ref{Gauss_compare}. 
\end{proof}

\begin{proof}[\textbf{Proof of Theorem \ref{f_Gauss_approx}}]
The asymptotic linear representation (\ref{k_approx}) yields that 
\begin{align*}
U_{n} &:= \max_{1 \leq \ell \leq N}\left|{\sqrt{n}h(\hat{k}_{\sharp}(x_{\ell}) - k_{\sharp}(x_{\ell})) \over \sigma_{n}(x_{\ell})}\right| = \max_{1 \leq \ell \leq N}|W_{n}(x_{\ell})| + o_{P}((\log n)^{-1/2})\\
&=:  V_{n} + o_{P}((\log n)^{-1/2}).
\end{align*}
This also implies that there exists a sequence of constants $\epsilon_{n} \downarrow 0$ such that
\[
P\left(|U_{n} - V_{n}| > \epsilon_{n}(\log n)^{-1/2}\right) \leq \epsilon_{n} 
\] 
(which follows from the fact that convergence in probability is metrized by the Ky Fan metric; see Theorem 9.2.2 in \cite{Du02}).
Then we have that 
\begin{align*}
P\left(U_{n} \leq t\right) &\leq P\left(\{U_{n} \leq t\} \cap \{|U_{n} - V_{n}| \leq \epsilon_{n}(\log n)^{-1/2}\}\right)\\
&\quad + P\left(\{U_{n} \leq t\} \cap \{|U_{n} - V_{n}| > \epsilon_{n}(\log n)^{-1/2}\}\right)\\
&\leq P\left(V_{n} \leq t + \epsilon_{n}(\log n)^{-1/2}\right) + \epsilon_{n}
\end{align*}
for any $t \in \mathbb{R}$. Theorem \ref{Gauss_approx2} yields that there exists a sequence of constants $\tilde{\epsilon}_{n} \downarrow 0$ such that
\begin{align*}
P\left(V_{n} \leq t + \epsilon_{n}(\log n)^{-1/2}\right) & \leq P\left(G_{n} \leq t + \epsilon_{n}(\log n)^{-1/2}\right) + \tilde{\epsilon}_{n}
\end{align*} 
for any $t \in \mathbb{R}$ where $G_{n} = \max_{1 \leq \ell \leq N}|Y_{\ell}|$. From the anti-concentration inequality for the maxima of Gaussian random vector (Theorem 3 in \cite{ChChKa15}), the right hand side is bounded from above by $P\left(G_{n} \leq t\right) + 8\epsilon_{n}(\log n)^{-1/2}E[G_{n}] + \tilde{\epsilon}_{n}$. 
Since $E[G_{n}] \leq D\log n$ for some positive constant $D$ which does not depend on $n$,
we have that 
\begin{align}\label{U_approx1}
P\left(U_{n} \leq t\right) &\leq P\left(G_{n} \leq t \right) + 9D\epsilon_{n} + \tilde{\epsilon}_{n} = P\left(G_{n} \leq t \right) + o(1)
\end{align}
for any $t \in \mathbb{R}$.
We also have that
\begin{align*}
P\left(V_{n} \leq t - \epsilon_{n}(\log n)^{-1/2}\right) &\leq P\left(\{V_{n} \leq t - \epsilon_{n}(\log n)^{-1/2}\} \cap \{|U_{n} - V_{n}| \leq \epsilon_{n}(\log n)^{-1/2}\}\right)\\
&\quad + P\left(\{V_{n} \leq t - \epsilon_{n}(\log n)^{-1/2}\} \cap \{|U_{n} - V_{n}| > \epsilon_{n}(\log n)^{-1/2}\}\right)\\
&\leq P\left(U_{n} \leq t \right) + \epsilon_{n}
\end{align*}
for any $t \in \mathbb{R}$. Therefore, we can show that 
\begin{align}\label{U_approx2}
P\left(U_{n} \leq t\right) &\geq P\left(G_{n} \leq t \right) - 9D\epsilon_{n} - \tilde{\epsilon}_{n} = P\left(G_{n} \leq t \right) + o(1)
\end{align}
for any $t \in \mathbb{R}$. Combining (\ref{U_approx1}) with (\ref{U_approx2}), we obtain the desired result. 
\end{proof} 

\section{On asymptotic validity of confidence bands}

We use the notations used in the proof of Theorem \ref{f_Gauss_approx} here. Let $q_{\tau}^{U_{n}}$ denotes the $(1-\tau)$-quantile of $U_{n}$. Theorem \ref{f_Gauss_approx} implies that there exists a sequence $\epsilon'_{n} \downarrow 0$ such that 
\[
\sup_{t \in \mathbb{R}}\left|P\left(U_{n} \leq t\right) - P\left(G_{n} \leq t\right)\right| \leq \epsilon'_{n}. 
\]
Then we have that
\begin{align*}
P\left(U_{n} \leq q_{\tau-\epsilon'_{n}}\right) &\geq P\left(G_{n} \leq q_{\tau - \epsilon'_{n}}\right) - \epsilon'_{n} = 1-\tau,
\end{align*}
where the last inequality holds $G_{n}$ has continuous distribution from the anti-concentration inequality (see Theorem 3 in \cite{ChChKa15}). This yields the inequality $q_{\tau}^{U_{n}} \leq q_{\tau-\epsilon'_{n}}$. Therefore, we have that 
\begin{align*}
P\left(U_{n} \leq q_{\tau}\right) &\leq P\left(U_{n} \leq q_{\tau-\epsilon'_{n}}\right)\\
&\leq P\left(G_{n} \leq q_{\tau - \epsilon'_{n}}\right) + \epsilon'_{n} = 1 - \tau + 2\epsilon'_{n}.
\end{align*}
Likewise, we have the inequality $q_{\tau+\epsilon'_{n}} \leq q_{\tau}^{U_{n}}$. This yields that 
\[
P\left(U_{n} \leq q_{\tau}\right) \geq 1-\tau -2\epsilon'_{n}. 
\]
Then we obtain $P\left(U_{n} \leq q_{\tau}\right) \to 1-\tau$ as $n \to \infty$.

\end{document}